\documentclass[preprint,12pt]{elsarticle}

\usepackage{amssymb, color}
\usepackage{graphicx}
\usepackage{amsmath}
\usepackage{amsthm}
\usepackage{algorithm}
\usepackage{algorithmic}
\usepackage{enumerate}
\usepackage{array}
\usepackage{multirow}
\usepackage{natbib}
\newtheorem{theorem}{Theorem}
\newtheorem{definition}{Definition}
\newtheorem{*proof}{Proof}

\journal{Information Sciences}

\begin{document}

\begin{frontmatter}



\title{Geometric transformations of multidimensional color images based on NASS}
%{Multidimensional color image geometric transformations based  on NASS}


\author[label1,label2]{Ping Fan}
\author[label3,label1]{Ri-Gui Zhou\corref{cor1}}
\ead{zrg@ecjtu.jx.cn}
\cortext[cor1]{Corresponding author.}
\author[label2]{Naihuan Jing}
\author[label1]{Hai-Sheng Li}
\address[label1]{School of Information Engineering, East China Jiaotong University, Nanchang, Jiangxi 330013, P.R. China}
\address[label2]{Department of mathematics, North Carolina State University, Raleigh, NC 27695, USA}
\address[label3]{College of Information Engineering, Shanghai Maritime University, Shanghai, 201306, P.R. China}

\begin{abstract}
%In this study,
We present quantum algorithms to realize
geometric transformations (two-point swappings, symmetric flips, local flips,
orthogonal rotations, and translations) based on an $n$-qubit normal arbitrary superposition state (NASS).
These transformations are implemented using quantum circuits consisting of basic quantum gates,
which are constructed with polynomial numbers of single-qubit and two-qubit gates.
Complexity analysis shows that the global operators (symmetric flips, local flips, orthogonal rotations)
 can be implemented with  $O(n)$ gates.
The proposed geometric transformations %These proposed geometric transformation operators be used
are used to facilitate applications of quantum images with low complexity.
\end{abstract}


\end{frontmatter}

% \linenumbers

\section{Introduction}
\label{sec:level1}
Quantum computing is a %the
combination of quantum mechanics and computer science. As a new computing paradigm, it offers %becomes
one possible solution to the challenge %problem
posed by the failure of Moore's law \cite{bibitem13}. Quantum computing \cite{bibitem1} has several unique computational features such as quantum coherence, entanglement, and superposition of quantum states,
which make quantum computing superior to its classical counterpart
in information storage and parallel computing \cite{bibitem1}.
%It has rapidly become an international research focus.
Quantum algorithms such as the quantum Fourier transform \cite{bibitem13,bibitem2} are more efficient than their classical counterparts, as the quantum Fourier transform on ${2^n}$ elements can be performed with $O({n^2})$ gates,
while the best classical analog, the Fast Fourier Transform(FFT), needs $O(n{2^n})$ gates to implement \cite{bibitem13}. Other famous quantum algorithms such as Shor's discrete
logarithms and integer-factoring algorithms %in polynomial time
\cite{bibitem2}, Deutsch's parallel computing algorithm with quantum parallelism and coherence \cite{bibitem3}, and Grover's quadratic speed-up for an unordered database search algorithm \cite{bibitem4,bibitem45} have further shown the advantage of quantum computing over its classical counterpart.
It has been an important problem to come up with quantum algorithms for known classical ones.
%In fact, utilizing the unique properties of Shor's  are
%important problems for known classical algorithms. Nevertheless, it is an important problem to

One reason that a quantum system is superior to a classical computer in information storage is due to
the fact that the amplitude or phase of a quantum state can be used to store information \cite{bibitem11}. In fact, if we consider a system of n qubits that stores ${2^n}$ complex numbers, for $n=500$, ${2^{500}}$ is larger than the estimated number of atoms in the whole universe. Thus, trying to store all of these complex numbers would not be possible using any conceivable computer \cite{bibitem13}. In a quantum system, if the frequency of the physical nature of a color represents its color instead of the RGB model or the HIS model \cite{bibitema4,bibitema5,bibitema53}, a color can be stored using only a 1-qubit quantum state \cite{bibitem5} and an image may be stored as a quantum array \cite{bibitem5,bibitem6}. A flexible representation of quantum image (FRQI) state stores the colors and coordinates of a two-dimensional (2D) gray image with ${2^n}$ pixels using $(n+1)$ qubits \cite{bibitem7}.  Maximally entangled qubits can be used to represent the vertices of polygons, so images can be reconstructed without using any additional information \cite{bibitem6}. Information storage and retrieval were achieved based on the quantum amplitude in previous studies, but they can be also implemented based on the quantum phase \cite{bibitem11}.

Quantum computing is implemented via quantum gates. Universal quantum gates are expressed as combinations of single-qubit and two-qubit gates \cite{bibitem13,bibitem14}, where two-qubit gates are universal in quantum computing \cite{bibitem15}. In \cite{bibitem16}, an efficient scheme has been proposed for initializing a quantum register with an arbitrary superposed state and the application of the scheme to three specific cases was also discussed.

Many applications in both 2D and three-dimensional (3D) biomedical imaging require efficient techniques for geometric transformations of images \cite{bibitema54,bibitem17}. A polynomial interpolator structure is used for the high-quality geometric transformation of 2D and 3D images in a classical computer system \cite{bibitem17}. In a quantum system, linear transformations \cite{bibitem8,bibitem9,bibitem19}, including two-point swappings, flips, orthogonal rotations, and restricted geometric transformations, are applied to 2D images based on their FRQI state. In addition, secure and efficient image processings for quantum computers have been described \cite{bibitem21,bibitem28},
and these include quantum watermarking \cite{bibitem9}, quantum image encryption and decryption algorithms \cite{bibitema52,bibitem42,bibitema51,bibitem22}.

 An $n$-qubit normal arbitrary superposition state (NASS) can represent a k-dimensional color image (including pixels and colors) where $n$ qubits encode the colors and coordinates of ${2^n}$ pixels (e.g., a five-dimensional color image of $1024 \times 1024 \times 1024 \times 1024 \times 1024$ using only 50 qubits) \cite{bibitem25,bibitem40}.  Based on  NASS and the color treatment strategy described in \cite{bibitem10},  we propose
 a new scheme to implement geometric transformations for multidimensional color images, including two-point swappings, symmetric flips, local flips, orthogonal rotations, and translations. The quantum circuit of a two-point swapping is designed using Gray codes. The quantum circuits of other geometric transformations are implemented in terms of two-point swapping circuits.

The remainder of this paper is organized as follows. Some basic quantum gate operations and a k-dimensional color image representation are described in Section 2.  Geometric transformations for multidimensional color images are discussed in detail in Section 3. The simulated experiments are designed
in Section 4. The conclusions are
given in Section 5. %new

\section{\label{sec:level2} Basic quantum gates and representation of k-D color images}

\subsection{\label{sec:level30}Quantum circuits}
By the evolution postulate of quantum mechanics \cite{bibitem13}, the evolution of a closed quantum
system is a reversible process. There are several well-known models of reversible quantum systems, and
the notable ones are the quantum Turing machine %model
\cite{bibitem3},  quantum circuit model
 \cite{bibitema1}, quantum cellular automation %model
 \cite{bibitema2} and so on. These %reversible quantum system
 models are essentially equivalent \cite{bibitema3}, so we choose the easiest quantum circuit model to
 describe geometric transformations. %to as it is the easiest to understand.
 %So we select quantum circuit to achieve geometric transformations.
 A quantum circuit consists of a series of quantum gates to realize a specific function.
 %, which can realize a specific function.
 Some quantum circuits  %composed of a quantum gate
 are shown in Figure \ref{fig:1}. These circuits are executed from left-to-right,
 and each line in the circuits represents a wire.
 Equivalently a quantum circuit is a realization of a unitary matrix.

\subsection{\label{sec:level3}Basic quantum gates}
The state of a quantum system is described by a unit vector called {\it ket} in a Hilbert space. %which is called a ket by Dirac.
The left and right kets are
denoted by  $\left\langle \ \right|$ and  $\left| \  \right\rangle $ respectively \cite{bibitem13}.
%are the symbols used for right ket and left ket, respectively .

Let $\left| u \right\rangle$ and $\left| v \right\rangle$ be arbitrary two states %\cite{bibitem13}
given by
\[\left| u \right\rangle  = {\left[ {\begin{array}{*{20}{c}}
{{u_0}}& \cdots &{{u_{n - 1}}}
\end{array}} \right]^T}\]
and
\[\left| v \right\rangle  = {\left[ {\begin{array}{*{20}{c}}
{{v_0}}& \cdots &{{v_{n - 1}}}
\end{array}} \right]^T}\]
where ${\left[ \cdot \right]^T}$ denotes the matrix transpose and ${u_i},{v_i} \in\mathbb C$, the set
of complex numbers. % and $i = 0,1, \cdots, n - 1$.

The Hermitian conjugates of $\left| u \right\rangle$ and $\left| v \right\rangle$ are denoted by
\[\left\langle u \right| = {\left| u \right\rangle ^ + } = \left[ {\begin{array}{*{20}{c}}
{u_0^ + }& \cdots &{u_{n - 1}^ + }
\end{array}} \right]\]
and
\[\left\langle v \right| = {\left| v \right\rangle ^ + } = \left[ {\begin{array}{*{20}{c}}
{v_0^ + }& \cdots &{v_{n - 1}^ + }
\end{array}} \right]\]
where $u_i^+$ and $v_i^+$ are the complex conjugates. %${\left[ \cdot \right]^+}$ is a matrix conjugate transpose.

The special states $\left| {\rm{0}} \right\rangle$ and $\left| {\rm{1}} \right\rangle$ are defined as
\[\left| {\rm{0}} \right\rangle {\rm{ = }}\left[ {\begin{array}{*{20}{c}}
{\rm{1}}\\
{\rm{0}}
\end{array}} \right],  \ \left| {\rm{1}} \right\rangle {\rm{ = }}\left[ {\begin{array}{*{20}{c}}
{\rm{0}}\\
{\rm{1}}
\end{array}} \right].\]

The tensor product of column matrices is defined by
%The symbol $ \otimes $ denotes the tensor product of column matrices given as follows.
\[\begin{array}{l}
\left| u \right\rangle  \otimes \left| v \right\rangle = {\left[ {\begin{array}{*{20}{c}}
{{u_0}{v_0}}& \cdots &{{u_0}{v_{n - 1}}}& \cdots &{{u_{n - 1}}{v_0}}& \cdots &{{u_{n - 1}}{v_{n - 1}}}
\end{array}} \right]^T}
\end{array}\]
which is also written simply as %$\left| u \right\rangle \otimes \left| v \right\rangle $ can also be written as
$\left| u \right\rangle \left| v \right\rangle $  or $\left| {uv} \right\rangle $.

The $n$ fold tensor product $U \otimes U\otimes \cdots \otimes U$ of the operator $U$ is abbreviated as ${U^{ \otimes n}}$. Similarly $\left| u \right\rangle \otimes \left| u \right\rangle \cdots \otimes \left| u \right\rangle $ is expressed as ${\left| u \right\rangle ^{ \otimes n}}$.

A linear operator $A$ is unitary if $A{A^ + } = {A^ + }A = I$, where
${A^ + }$ is the conjugate transpose of $A$ and $I$ is the identity operator shown in Figure \ref{fig:1}.

\begin{figure}[!t]
\centering
\includegraphics[width=3in]{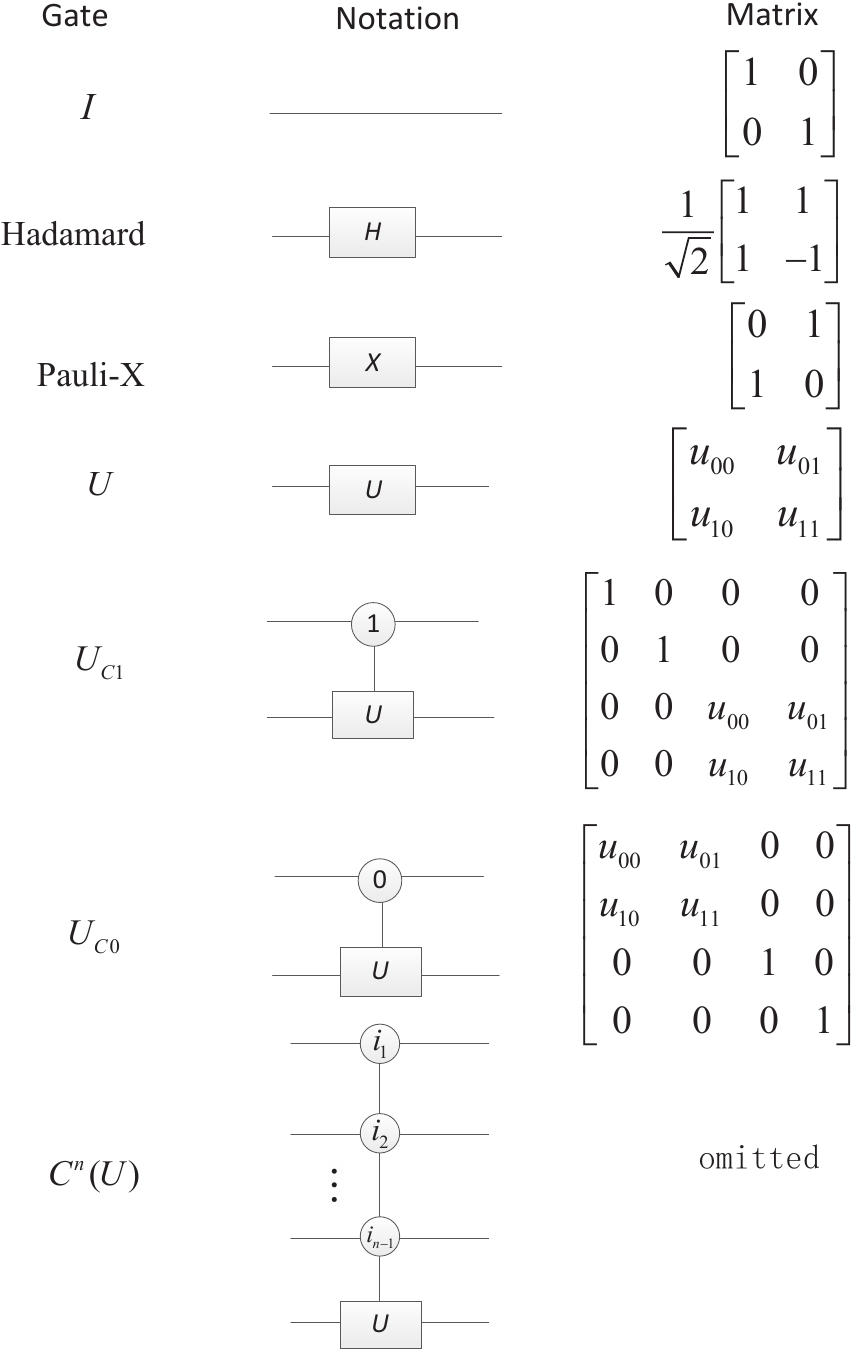}
 %where an .eps filename suffix will be assumed under latex,
 %and a .pdf suffix will be assumed for pdflatex; or what has been declared
 %via \DeclareGraphicsExtensions.
\caption{Notations for some basic gates with their corresponding matrices. The operator $U$ is unitary and the gate ${C^n}(U)$ is associated with ${i_1},{i_2}, \ldots ,{i_{n - 1}} \in \{ 0,1\}$.}
\label{fig:1}
\end{figure}

Some of the basic gates and their corresponding matrices are shown in Figure \ref{fig:1}.
The identity ($I$), Hadamard ($H$), and Pauli-X ($X$) gates are well-known and can be found in
the basic reference \cite{bibitem13}.
The gates $U$, ${U_{C1}}$,  ${U_{C0}}$, and ${C^n}(U)$ are
the fourth, fifth, sixth, and seventh gates in Figure \ref{fig:1}. They are defined as follows.

Let $\left| \psi \right\rangle$ and $\left| \xi \right\rangle$ be two arbitrary states
\begin{equation}
\label{eqn_315x1}
\left\{ \begin{array}{l}
\left| \psi  \right\rangle  = \alpha \left| 0 \right\rangle  + \beta \left| 1 \right\rangle \\
\left| \xi  \right\rangle  = c\left| 0 \right\rangle  + d\left| 1 \right\rangle
\end{array} \right., \qquad \alpha, \beta, c, d\in\mathbb C
\end{equation}
%where $\alpha$, $\beta$, $c$, and $d$ are complex numbers.
Then the action of $U$ is given by
%By applying gate $U$ to the state $\left| \psi \right\rangle$ in (\ref{eqn_315x1}), we have
\begin{equation}
\label{eqn_315x2}
U\left| \psi  \right\rangle  = \left( {\alpha {u_{00}} + \beta {u_{01}}} \right)\left| 0 \right\rangle  + \left( {\alpha {u_{10}} + \beta {u_{11}}} \right)\left| 1 \right\rangle
\end{equation}

When $U$ is the Pauli-X gate, Eq. (\ref{eqn_315x2}) becomes
\begin{equation}
\label{eqn_315x6}
X\left| \psi  \right\rangle  = X(\alpha \left| 0 \right\rangle  + \beta \left| 1 \right\rangle ) = \beta \left| 0 \right\rangle  + \alpha \left| 1 \right\rangle
\end{equation}

Apply the ${U_{C1}}$ and ${U_{C0}}$ gates to the states $\left| \psi \right\rangle$ and $\left| \xi \right\rangle $ in Eq. (\ref{eqn_315x1}), respectively:
\begin{equation}
\label{eqn_315x3}
{U_{C1}}{\rm{(}}\left| \xi  \right\rangle \left| \psi  \right\rangle ) = c\left| 0 \right\rangle \left| \psi  \right\rangle  + d\left| 1 \right\rangle (U\left| \psi  \right\rangle )
\end{equation}
and
\begin{equation}
\label{eqn_315x4}
{U_{C0}}{\rm{(}}\left| \xi  \right\rangle \left| \psi  \right\rangle ) = c\left| 0 \right\rangle (U\left| \psi  \right\rangle ) + d\left| 1 \right\rangle \left| \psi  \right\rangle
\end{equation}
where $U\left| \psi \right\rangle$ is shown in Eq. (\ref{eqn_315x2}).

When $U$ is the Pauli-X gate, %in Figure \ref{fig:1},
then ${U_{C1}}$ and ${U_{C0}}$ are designated as ${N_{C1}}$ and
${N_C}_0$, respectively, which are also called the controlled-NOT, as shown in Figure \ref{fig:2}.
%Apply the ${N_{C1}}$ and ${N_C}_0$ gates to the states $\left| \psi \right\rangle$ and $\left| \xi \right\rangle $ in (\ref{eqn_315x1}), respectively:
Therefore,
\begin{equation}
\label{eqn_315x7}
{N_{C1}}{\rm{(}}\left| \phi  \right\rangle \left| \psi  \right\rangle ) = c\left| 0 \right\rangle \left| \psi  \right\rangle  + d\left| 1 \right\rangle (X\left| \psi  \right\rangle )
\end{equation}
and
\begin{equation}
\label{eqn_315x8}
{N_{C0}}{\rm{(}}\left| \phi  \right\rangle \left| \psi  \right\rangle ) = c\left| 0 \right\rangle (X\left| \psi  \right\rangle ) + d\left| 1 \right\rangle \left| \psi  \right\rangle
\end{equation}
%where $X\left| \psi \right\rangle$ is shown in (\ref{eqn_315x6}).

The ${C^n}(U)$ ($n \ge 3$) gate associated with $i_1, \cdots, i_{n-1}\in \{ 0,1\}$ acts on the $n$-fold tensor product of
the Hilbert space. Its action on
%By applying the ${C^n}(U)$ ($n \ge 3$) gate to
the tensor product of $n$ one-qubit
states $\left| {{x_{\rm{1}}}{x_2} \cdots {x_{n - 1}}} \right\rangle \left| \psi \right\rangle$ is given by %we obtain
\begin{equation}
\label{eqn_315x5}
\begin{array}{l}
{C^n}(U)(\left| {{x_1}{x_2} \cdots {x_{n - 1}}} \right\rangle \left| \psi  \right\rangle )\\
 = \left| {{x_1} \cdots {x_{n - 1}}} \right\rangle {U^{f({x_1} \cdots {x_{n - 1}},{i_1} \cdots {i_{n - 1}})}}\left| \psi  \right\rangle
\end{array}
\end{equation}
where %${x_1}, \ldots ,{x_{n - 1}},{i_1}, \ldots ,{i_{n - 1}} \in \{ 0,1\}$ (see Figure \ref{fig:1}),
%${i_1},{i_2}, \cdots {i_{n - 1}}$ are the numbers in the ${C^n}(U)$ gate,
the function $f=f(x_1\cdots x_{n - 1}, i_1 \cdots i_{n - 1})$ is defined by
\begin{equation}\label{e:function}
\left\{ {\begin{array}{*{20}{c}}
{{f({x_1} \cdots {x_{n - 1}}, {i_1} \cdots {i_{n - 1}})}= 0,}&{{x_1} \cdots {x_{n - 1}}{\rm{ = }}{i_1} \cdots {i_{n - 1}}}\\
{{f({x_1} \cdots {x_{n - 1}}, {i_1} \cdots {i_{n - 1}})}= 1,}&{{x_1} \cdots {x_{n - 1}} \ne {i_1} \cdots {i_{n - 1}}}
\end{array}} \right.
\end{equation}
%\[\left\{ {\begin{array}{*{20}{c}}
%{{U^{f({x_1} \cdots {x_{n - 1}},{i_1} \cdots {i_{n - 1}})}}= I,}&{{x_1} \cdots {x_{n - 1}}{\rm{ = }}{i_1} \cdots {i_{n - 1}}}\\
%{{U^{f({x_1} \cdots {x_{n - 1}},{i_1} \cdots {i_{n - 1}})}}= U,}&{{x_1} \cdots {x_{n - 1}} \ne {i_1} \cdots {i_{n - 1}}}
%\end{array}} \right.\]

For $k=1, \cdots, n$, one also defines the gate $C^n(U_k)$ as the composition of $C^n(U)$
and the permutation operator $P_{k, n}$, which switches the $k$th and $n$th factors of arbitrary tensor
product of $n$ one-qubits:
\begin{equation*}
P_{k, n}|j_1 \cdots j_k\cdots j_n\rangle=|j_1 \cdots j_n\cdots j_k\rangle.
\end{equation*}

In particular, for the Pauli X-gate, %If $U = X$ and the $U$ gate is located in the $kth$ qubit of the ${C^n}(U)$ gate,
the gate ${C^n}({X_k})$
($k = 1,2, \cdots ,n$) associated with $i_1, \cdots, i_{n-1}\in\{0, 1\}$ is shown in Figure \ref{fig:x2}. Its exact action on the tensor product $\left| {{j_{\rm{1}}} \cdots {j_{k - 1}}{j_k}{j_{k + 1}} \cdots {j_n}} \right\rangle $
is given by
\begin{equation}
\label{eqn_xx1}
\begin{array}{l}
{C^n}({X_k})\left| {{j_{\rm{1}}} \cdots {j_{k - 1}}{j_k}{j_{k + 1}} \cdots {j_n}} \right\rangle \\
 = \left| {{j_{\rm{1}}} \cdots {j_{k - 1}}} \right\rangle ({X^f}\left| {{j_k}} \right\rangle )\left| {{j_{k + 1}} \cdots {j_n}} \right\rangle
\end{array}
\end{equation}
where $f=f(j_1\cdots j_{k-1}j_{k+1}\cdots j_n, i_1\cdots i_{n-1})$ is defined in Eq. (\ref{e:function}).
%${{j_{\rm{1}}} \cdots {j_{n}}} \in \{ 0,1\}$
%${i_1},{i_2}, \cdots {i_{n - 1}}$ are the numbers in the ${C^n}({X_k})$ gate (see Figure \ref{fig:x2}),
%and $X^f$ is defined by
%\[\left\{ {\begin{array}{*{20}{c}}
%{{X^f} = I,}&{{j_1} \cdots {j_{k - 1}}{j_{k + 1}} \cdots {j_n} \ne {i_1} \cdots {i_{k - 1}}{i_{k}} \cdots {i_{n-1}}}\\
%{{X^f} = X,}&{{j_1} \cdots {j_{k - 1}}{j_{k + 1}} \cdots {j_n}{\rm{ = }}{i_1} \cdots {i_{k - 1}}{i_{k}} \cdots {i_{n-1}}}
%%{{X^f} = I,}&{{j_1} \cdots {j_{k - 1}}{j_{k + 1}} \cdots {j_n} \ne {i_1} \cdots {i_{k - 1}}{i_{k + 1}} \cdots {i_n}}\\
%%{{X^f} = X,}&{{j_1} \cdots {j_{k - 1}}{j_{k + 1}} \cdots {j_n}{\rm{ = }}{i_1} \cdots {i_{k - 1}}{i_{k + 1}} \cdots {i_n}}
%\end{array}} \right.\]
For example, when $n = 3$ and ${i_1} = {i_3} = 1$, we have ${C^3}({X_2})\left| {{\rm{111}}} \right\rangle {\rm{ = }}\left| {{\rm{101}}} \right\rangle $.

\begin{figure}[!t]
\centering
\includegraphics[width=2in]{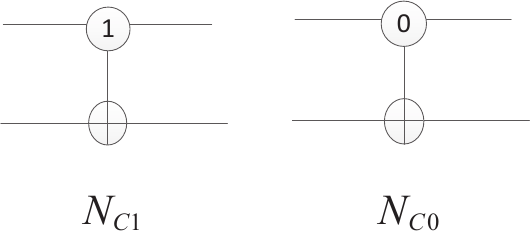}
 %where an .eps filename suffix will be assumed under latex,
 %and a .pdf suffix will be assumed for pdflatex; or what has been declared
 %via \DeclareGraphicsExtensions.
\caption{Two controlled-NOT gates.}
\label{fig:2}
\end{figure}
\begin{figure}[!t]
\centering
\includegraphics[width=3in]{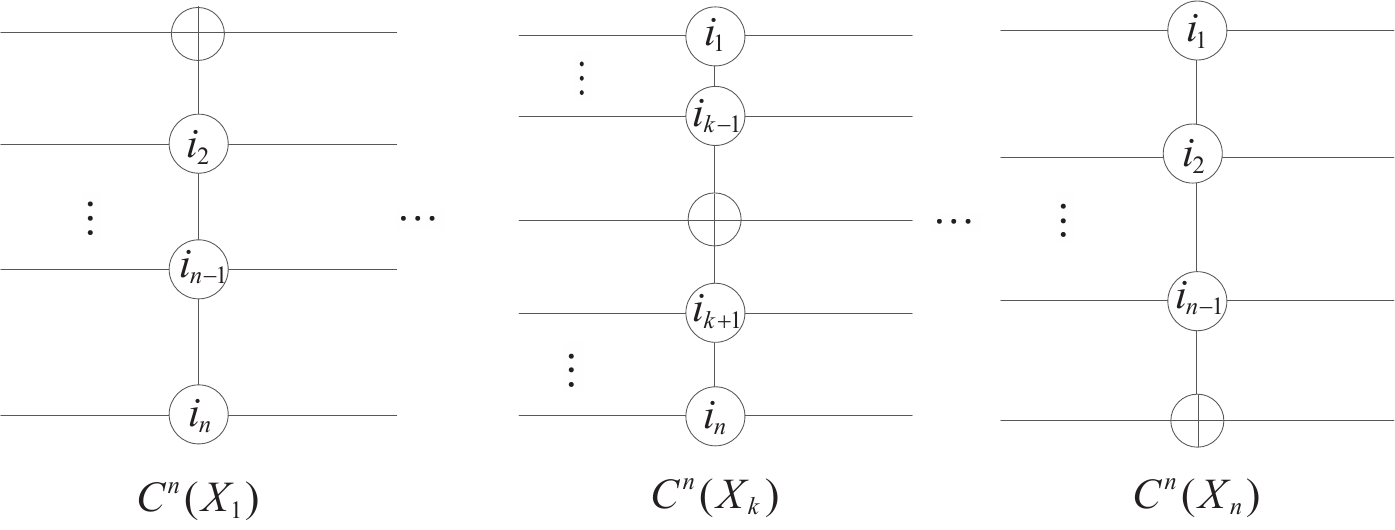}
 %where an .eps filename suffix will be assumed under latex,
 %and a .pdf suffix will be assumed for pdflatex; or what has been declared
 %via \DeclareGraphicsExtensions.
\caption{Set of $n_k$-qubit controlled-NOT gates ($k = 1,2, \cdots ,n$).}
\label{fig:x2}
\end{figure}

\subsection{\label{sec:level4}Representation of a quantum image}
We recall how NASS represents a multidimensional color image \cite{bibitem40}.

Suppose there are $M$ colors and $Color = \{ colo{r_1},colo{r_2}, \cdots ,colo{r_M}\}$ is
a sorted set of the colors.
We convert the $i$th color to an angle value by
\begin{equation}\label{eqn_x1}
{\phi _i} = \frac{{\pi (i - 1)}}{{2(M - 1)}}
\end{equation}
where $1\leqslant i\leqslant M$. This establishes a one-to-one correspondence ${F_1}$ between the sorted color set and a sorted angle set $\phi=\{\phi_1, \phi_2, \cdots, \phi_M\}$.
For RGB color images, $M = {2^{24}}$, let $x, y, z$ be the values of $ R,G, B$ in 24-bit RGB True Color, and let $i = x \times 256 \times 256 + y \times 256 + z + 1$. Thus, $colo{r_i}$ corresponds to the value ($x,y,z$) of $RGB$. For instance, $colo{r_1}$ and $colo{r_{16777216}}$ correspond to the RGB values (0,0,0) and (255,255,255) respectively.

By a k-dimensional color image, we mean a k-dimenisional lattice of
size $2^{m_1}\times 2^{m_2}\times \cdots \times 2^{m_k}$, where each lattice point is painted with a color.
For simplicity we often assume that the image has spread out the full range
of $0, 1, \cdots, 2^{m_i}-1$ in the ith axis, and if
the actual size is less, we can take the color to be empty in the remaining spots,
and then we enlarge the color set by the empty color.
We only consider the $n$-dimensional skeleton situated in the
positive direction of each axis.
Let $n=m_1+m_2+ \ldots+m_k\ $, then we assign the image a
quantum superposition state in the $2^n$-dimensional Hilbert
space:
%To represent a $k$-dimensional color image, a quantum superposition state in ${2^n}$ dimensional Hilbert space is assigned
 \begin{equation}
\label{eqn_2}
{\left| {{\psi _\phi }} \right\rangle _k}{\rm{ = }}\sum\limits_{i = 0}^{{2^n} - 1} {{a_i}\left| i \right\rangle } {\rm{ = }}\sum\limits_{i = 0}^{{2^n} - 1} {{a_i}\left| {{v_1}} \right\rangle } \left| {{v_2}} \right\rangle  \cdots \left| {{v_k}} \right\rangle \end{equation}
 where $\left| {{v_1}} \right\rangle ,\left| {{v_2}} \right\rangle , \ldots ,\left| {{v_k}} \right\rangle$ correspond
 to the $k$ axes of the image respectively and $a_i$ is the angle value at the site $i$.
The running index $i$ goes through all lattice points of the image, and
 %$i = {i_1} \cdots {i_j}{i_{j + 1}} \cdots {i_l} \cdots {i_m} \cdots {i_n}$, ${v_1} = {i_1} \cdots {i_j}$
 $i = {i_1} i_2\cdots {i_n}$
 is its binary expansion obtained by taking the superposition of $v_1, v_2, \cdots, v_k$, where ${v_1} = {i_1}\cdots {i_{m_1}}$, ${v_{\rm{2}}} = {i_{m_1 + 1}}\cdots {i_{m_1+m_2}}$ and ${v_k} = {i_{n-m_k+1} \cdots {i_n}}$
 are in their binary expansions. %of ${v_1},\ {v_2}$ and ${v_k}$ respectively.
 Thus
 $\left| i \right\rangle  = \left| {{v_1}} \right\rangle \left| {{v_2}} \right\rangle  \cdots \left| {{v_k}} \right\rangle $ is the coordinate $\left( {{v_1},{v_2}, \cdots ,{v_k}} \right)$ of $i$ in the k-dimensional space
 and ${a_i}\in \phi $ is the value of the color
 at site $i$. %that corresponds to the coordinate. {\color{red} def of $F_1$ is not clear}

We define the NASS state of the image to be the normalized state
%To normalize the state $\left| {{\psi _\phi }} \right\rangle $ in (\ref{eqn_2}), we set
%\begin{equation}
%\label{eqn_q3}
%{\theta _i} = {{{a_i}} \mathord{\left/
% {\vphantom {{{a_i}} {\sqrt {\sum\nolimits_{y = 0}^{{2^n} - 1} {a_y^2} } }}} \right.
% \kern-\nulldelimiterspace} {\sqrt {\sum\nolimits_{y = 0}^{{2^n} - 1} {a_y^2} } }}
%\end{equation}
%Substituting (\ref{eqn_q3}) for ${a_i}$ in (\ref{eqn_2}), we obtain an NASS state ${\left| \psi \right\rangle _k}$
\begin{equation}
\label{eqn_3}
{\left| \psi  \right\rangle _k}{\rm{ = }}\sum\limits_{i = 0}^{{2^n} - 1} {{\theta _i}\left| {{v_1}} \right\rangle } \left| {{v_2}} \right\rangle  \cdots \left| {{v_k}} \right\rangle
\end{equation}
where $\theta_i=a_i/\sqrt {\sum\nolimits_{y = 0}^{{2^n} - 1} {a_y^2} }$. Note that the magnitude square
$\sum\nolimits_{y = 0}^{{2^n} - 1} {a_y^2}$ is a constant for the image, so
$\theta_i$ can replace $a_i$ to represent the color at the site $i$.
In this way the NASS state ${\left| \psi \right\rangle _k} $ in Eq. (\ref{eqn_3}) represents a k-dimensional color image.

%Note that
%${\theta _{\rm{i}}}\sqrt {\sum\nolimits_{y = 0}^{{2^n} - 1} {a_y^2} } {\rm{ = }}{a_i}$ denotes
%the color at lattice point $i$, and $\sqrt {\sum\nolimits_{y = 0}^{{2^n} - 1} {a_y^2} }$ is a constant for the color
%image, thus we can use ${\theta _{\rm{i}}}$ and ${a_i}$ to denote the same color by creating a bijective function ${F_2}:{\theta _i} \leftrightarrow {a_i}$.

Each axis $|v_i\rangle$ spans a Hilbert space of dimension $2^{m_i}$. In view of its role in the image, %new
we abuse the notation to denote
$\dim(|v_i\rangle)=m_i$ and also call it the size of the state $|v_i\rangle$.
%\begin{definition}
%\label{def_1}
% $\dim (\left| {{u}} \right\rangle )$ is the size of the state $\left| {{u}} \right\rangle $.
%\end{definition}
e.g., $\dim (\left| {000} \right\rangle ) = 3$, and we have
 \begin{equation}
\label{eqn_316x4}
\dim ({\left| \psi \right\rangle _k}) = \sum\limits_{i = 1}^k {\dim (\left| {{v_i}} \right\rangle )} = n
\end{equation}
where ${\left| \psi \right\rangle _k}$ is the NASS state shown in Eq. (\ref{eqn_3}).

To make Eq. (\ref{eqn_3}) clearer,
let us consider $\dim (\left| {{v_i}} \right\rangle ) = {m_i}$, $i = 1,2, \ldots, k$ as an example. The NASS state representing a multidimensional color image with ${2^n}$ pixels is written explicitly as follows.
\begin{equation}
\label{eqn_316x5}
\begin{array}{l}
{\left| \psi  \right\rangle _k} {\rm{ = }}\sum\limits_{i = 0}^{{2^n} - 1} {{\theta _i}\left| {{i_1} \cdots {i_{{m_1}}}} \right\rangle }  \cdots \left| {{i_{(\sum\limits_{h = 1}^{j - 1} {{m_h}} ) + 1}} \cdots {i_{(\sum\limits_{h = 1}^{j - 1} {{m_h}} ) + {m_j}}}} \right\rangle \\
 \cdots \left| {{i_{(\sum\limits_{h = 1}^{k - 1} {{m_h}} ) + 1}} \cdots {i_{(\sum\limits_{h = 1}^{k} {{m_h}} )}}} \right\rangle
\end{array}
\end{equation}
where $1 \le j \le k$ and the binary expansion of the integer $i$ is
\begin{equation}
\label{eqn_316x6}
\begin{array}{l}
i = {i_1} \cdots {i_{{m_1}}} \cdots {i_{(\sum\limits_{h = 1}^{j - 1} {{m_h}} ) + 1}} \cdots {i_{(\sum\limits_{h = 1}^{j - 1} {{m_h}} ) + {m_j}}}\\
 \cdots {i_{(\sum\limits_{h = 1}^{k - 1} {{m_h}} ) + 1}} \cdots {i_{(\sum\limits_{h = 1}^{k - 1} {{m_h}} ) + {m_k}}}
\end{array}
\end{equation}

 Substituting $n = 5$, $k = 2$, ${m_1} = 2$ and ${m_{\rm{2}}} = {\rm{3}}$ into Eq. (\ref{eqn_316x5}), we obtain
\begin{equation}
\label{eqn_x5}
\begin{array}{l}
{\left| \psi  \right\rangle _2}{\rm{ = }}\sum\limits_{i = 0}^{{2^5} - 1} {{\theta _i}\left| {{i_1}{i_{\rm{2}}}} \right\rangle } \left| {{i_{\rm{3}}}{i_{\rm{4}}}{i_{\rm{5}}}} \right\rangle {\rm{ = }}{\theta _0}\left| {00} \right\rangle \left| {000} \right\rangle \\
 + {\theta _1}\left| {00} \right\rangle \left| {001} \right\rangle  +  \cdots  + {\theta _{30}}\left| {11} \right\rangle \left| {110} \right\rangle  + {\theta _{31}}\left| {11} \right\rangle \left| {111} \right\rangle
\end{array}
\end{equation}
%where ${\left| \psi \right\rangle_2}$ can represent the 2D color image shown in Figure \ref{fig:x4}.
which represents the 2D color image shown in Figure \ref{fig:x4}.

\begin{figure}[!h]
\centering
\includegraphics[width=2in]{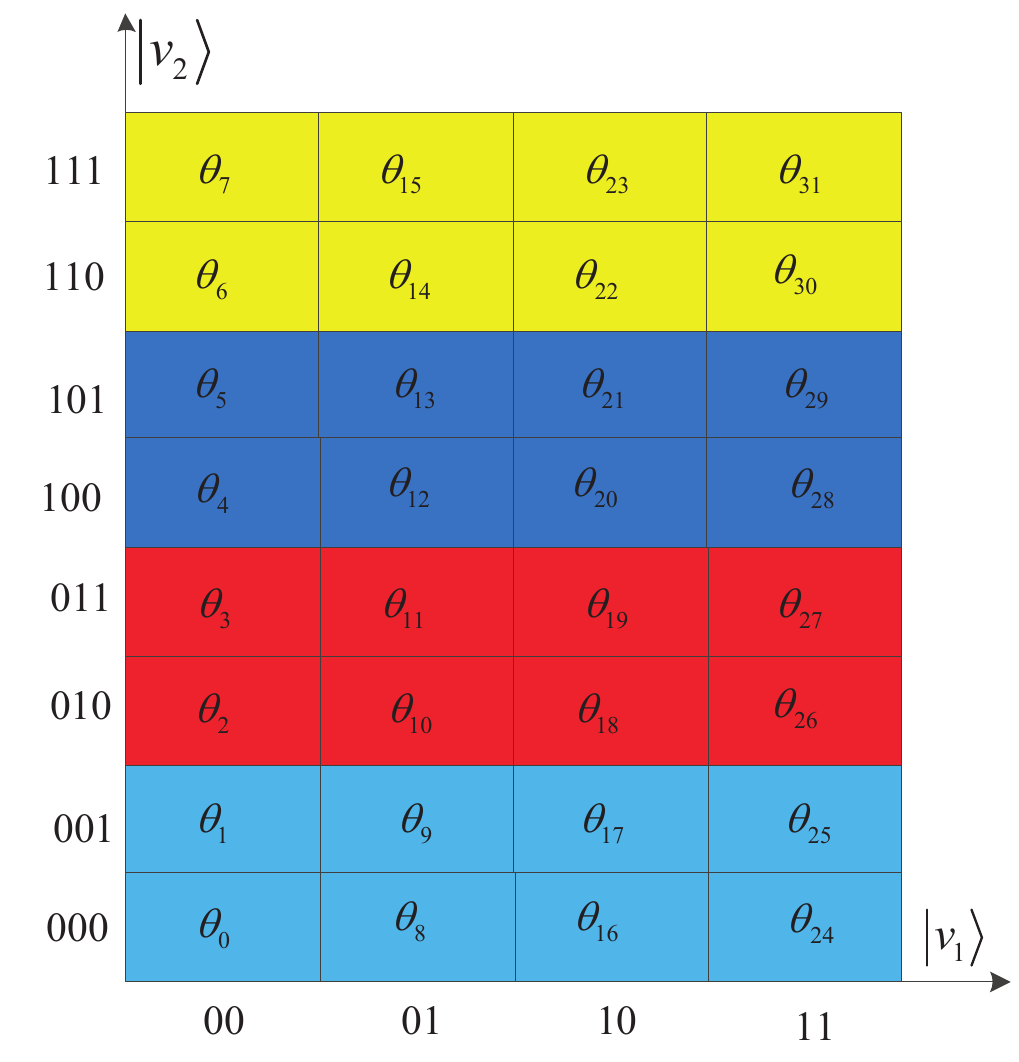}
 %where an .eps filename suffix will be assumed under latex,
 %and a .pdf suffix will be assumed for pdflatex; or what has been declared
 %via \DeclareGraphicsExtensions.
\caption{Two-dimensional color image of $4 \times 8$.}
\label{fig:x4}
\end{figure}

 Substituting $n = 5$, $k = 3$, ${m_1} = 2$, ${m_2} = 2$ and ${m_{\rm{3}}} = {\rm{1}}$ into Eq. (\ref{eqn_316x5}), we obtain
\begin{equation}
\label{eqn_x6}
\begin{array}{l}
{\left| \psi  \right\rangle _3}{\rm{ = }}\sum\limits_{i = 0}^{{2^5} - 1} {{\theta _i}\left| {{i_1}{i_{\rm{2}}}} \right\rangle } \left| {{i_{\rm{3}}}{i_{\rm{4}}}} \right\rangle \left| {{i_{\rm{5}}}} \right\rangle {\rm{ = }}{\theta _0}\left| {00} \right\rangle \left| {00} \right\rangle \left| 0 \right\rangle \\
 + {\theta _1}\left| {00} \right\rangle \left| {00} \right\rangle \left| 1 \right\rangle  +  \cdots  + {\theta _{30}}\left| {11} \right\rangle \left| {11} \right\rangle \left| 0 \right\rangle  + {\theta _{31}}\left| {11} \right\rangle \left| {11} \right\rangle \left| 1 \right\rangle
\end{array}
\end{equation}
%where ${\left| \psi \right\rangle_3}$ can represent the 3D color image shown in Figure \ref{fig:x5}.
which represents the 3D color image shown in Figure \ref{fig:x5}.

\begin{figure}[!h]
\centering
\includegraphics[width=3in]{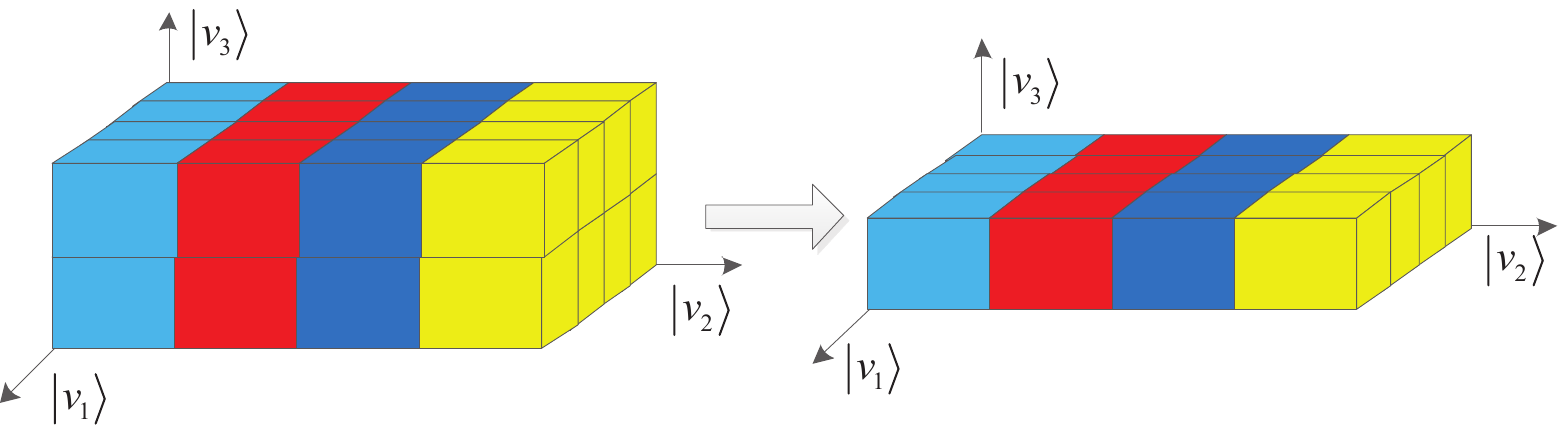}
 %where an .eps filename suffix will be assumed under latex,
 %and a .pdf suffix will be assumed for pdflatex; or what has been declared
 %via \DeclareGraphicsExtensions.
\caption{3D color image of $4 \times 4 \times 2$. In the lower sub-image on the right-hand side, $\left| {{i_5}} \right\rangle = \left| 0 \right\rangle $.}
\label{fig:x5}
\end{figure}

 Substituting $n = 5$, $k = 5$, and ${m_{\rm{1}}} = {m_{\rm{2}}}{\rm{ = }}{m_{\rm{3}}}{\rm{ = }}{m_{\rm{4}}}{\rm{ = }}{m_{\rm{5}}}{\rm{ = 1}}$ into Eq. (\ref{eqn_316x5}), we obtain
\begin{equation}
\label{eqn_x7}
\begin{array}{l}
\left| {{\psi _{\rm{5}}}} \right\rangle {\rm{ = }}\sum\limits_{i = 0}^{{2^5} - 1} {{\theta _i}\left| {{i_1}} \right\rangle } \left| {{i_{\rm{2}}}} \right\rangle \left| {{i_{\rm{3}}}} \right\rangle \left| {{i_{\rm{4}}}} \right\rangle \left| {{i_{\rm{5}}}} \right\rangle \\
 = {\theta _0}\left| 0 \right\rangle \left| 0 \right\rangle \left| 0 \right\rangle \left| 0 \right\rangle \left| 0 \right\rangle  + {\theta _1}\left| 0 \right\rangle \left| 0 \right\rangle \left| 0 \right\rangle \left| 0 \right\rangle \left| 1 \right\rangle \\
 +  \cdots  + {\theta _{30}}\left| 1 \right\rangle \left| 1 \right\rangle \left| 1 \right\rangle \left| 1 \right\rangle \left| 0 \right\rangle  + {\theta _{31}}\left| 1 \right\rangle \left| 1 \right\rangle \left| 1 \right\rangle \left| 1 \right\rangle \left| 1 \right\rangle
\end{array}
\end{equation}
which represents the five-dimensional (5D) color image shown in Figure \ref{fig:x6}.

\begin{figure}[!h]
\centering
\includegraphics[width=3.5in]{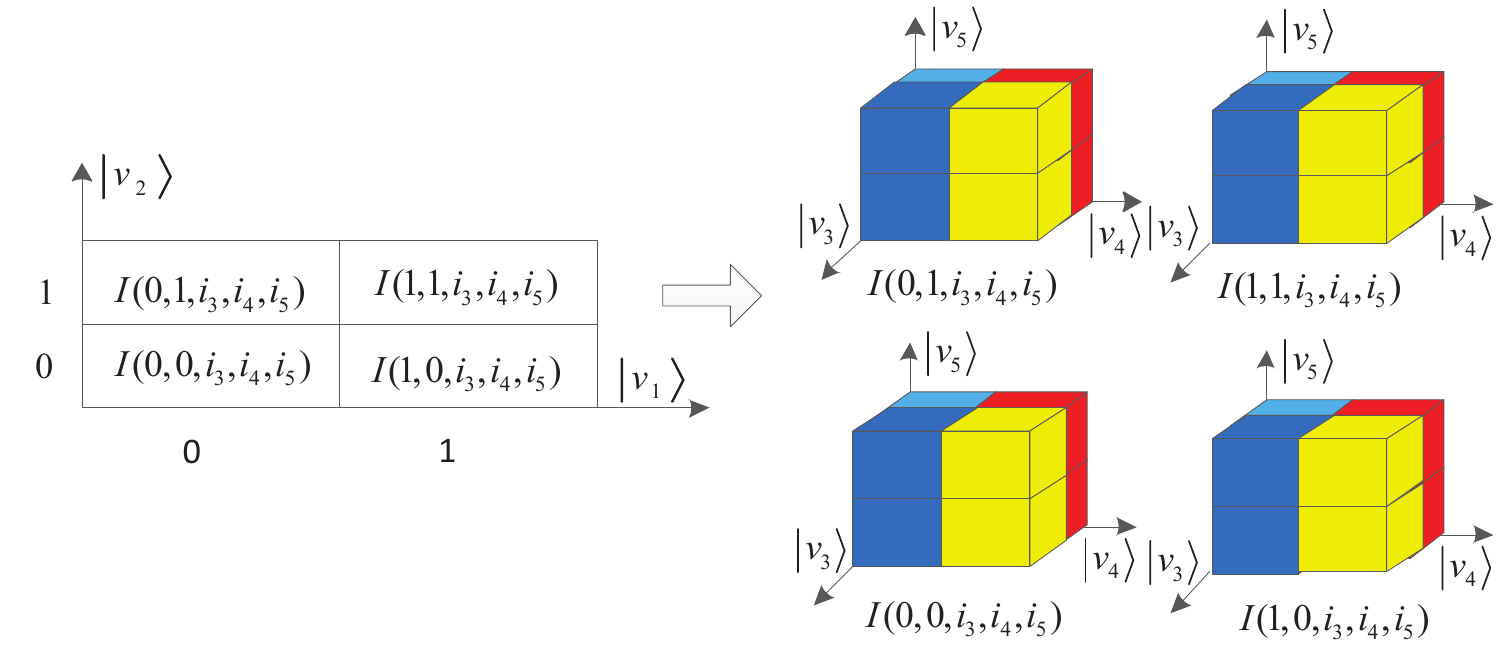}
 %where an .eps filename suffix will be assumed under latex,
 %and a .pdf suffix will be assumed for pdflatex; or what has been declared
 %via \DeclareGraphicsExtensions.
\caption{5D color image of $2 \times 2 \times 2 \times 2 \times 2$. The image is divided into four sub-images, i.e.,
$I(0,0,{i_3},{i_4},{i_5})$, $I(0,1,{i_3},{i_4},{i_5})$, $I(1,0,{i_3},{i_4},{i_5})$, and $I(1,1,{i_3},{i_4},{i_5})$. The projections of the four sub-images into the 3D space spanned by $\left| {{v_3}} \right\rangle \times \left| {{v_4}} \right\rangle \times \left| {{v_5}} \right\rangle $ are shown on the right.}
\label{fig:x6}
\end{figure}

\section{\label{sec:level6}Geometric transformations for multidimensional color images}
Geometric transformations on quantum images \cite{bibitem8}, such as two-point swappings, flips, and
orthogonal rotations, can be applied to 2D images based on the FRQI state. Using %a previous study
\cite{bibitem8,bibitem19}, we discuss
geometric transformations %(including two-point swapping, flip, orthogonal rotation, and translation transformations)
for multidimensional color images based on the NASS state.

\subsection{\label{sec:level61} Two-point swappings}
\begin{definition}
\label{def_2}
A two-point swapping %operator
${G_T}$ for a k-dimensional image  is the linear operator defined by
\begin{equation}
\label{eqn_11}
{G_T} = \left| s \right\rangle \left\langle t \right| + \left| t \right\rangle \left\langle s \right| + \sum\limits_{i = 0,i \ne s,t}^{{2^n} - 1} {\left| i \right\rangle \left\langle i \right|}
\end{equation}
where $\left| {\rm{s}} \right\rangle = \left| {v_1^s} \right\rangle \left| {v_2^s} \right\rangle \cdots \left| {v_k^s} \right\rangle $ and $\left| t \right\rangle = \left| {v_1^t} \right\rangle \left| {v_2^t} \right\rangle \cdots \left| {v_k^t} \right\rangle $ are the coordinates of the two swapped pixels, and $\left| i \right\rangle = \left| {{v_1}} \right\rangle \left| {{v_2}} \right\rangle \cdots \left| {{v_k}} \right\rangle $ run
through the coordinates of the other pixels.
%{\color{red} The binary expansions of the integers $s$, $t$, and $i$ are $s = {s_1} \cdots {s_n}$, $t = {t_1} \cdots {t_n}$, and $i = {i_1} \cdots {i_n}$, respectively. $v_j^s$ and $v_j^{\rm{t}}$($j = 1,2, \cdots k$) are the values of the $jth$ coordinate components of $s$ and $t$, respectively.}
%{\color{blue} How to match the dimension?}
\end{definition}

Since ${G_T}.G_T^ + = {I^{ \otimes n}}$, ${G_T}$ is a unitary operator.
Applying ${G_T}$ to the NASS state $\left| \psi \right\rangle_{k}$ in Eq. (\ref{eqn_3}), we get
\begin{equation}
\label{eqn_12}
\begin{array}{l}
{G_T}({\left| \psi  \right\rangle _k}) = \sum\limits_{i = 0}^{{2^n} - 1} {{\theta _i}{G_T}(\left| i \right\rangle } )\\
 = {\theta _s}\left| t \right\rangle  + {\theta _t}\left| s \right\rangle  + \sum\limits_{i = 0,i \ne s,t}^{{2^n} - 1} {{\theta _i}\left| i \right\rangle }
\end{array}
\end{equation}
That is, ${G_T}$ swaps two colors of a $k$-dimensional color image. %from (\ref{eqn_12}).
We can use
Gray codes \cite{bibitem13} to design a quantum circuit for the operator ${G_T}$ of
two-point swapping. Suppose that $s$ and $t$ are two distinct binary numbers, then a Gray code
connecting $s$ and $t$ is a sequence of binary numbers that starts with $s$ and ends with $t$
such that adjacent binary numbers in the list differ by exactly one bit. For example, a Gray code connecting
$s = 0 \cdots 0 \cdots 0$ and $t = 1 \cdots 1 \cdots 1$ is
\begin{equation}
\label{eqn_318x1}
\begin{array}{*{20}{c}}
0& \cdots &0& \cdots &0\\
0& \cdots &0& \cdots &1\\
 \vdots &{}& \vdots &{}& \vdots \\
0& \cdots &1& \cdots &1\\
 \vdots &{}& \vdots &{}& \vdots \\
1& \cdots &1& \cdots &1
\end{array}
\end{equation}

Suppose $s$ and $t$ are two binary numbers with at most $n$ digits different.
Then we can find a Gray code $g_1, g_2, \ldots, g_m$ to connect $s$ and $t$
with $m (\le n+1)$ elements,
%Let ${g_1},{g_2}, \ldots {g_m}$ be the elements of a Gray code that connects $s$ and $t$,
where ${g_1} = s$ and ${g_m} = t$.
%$s$ and $t$ differ in at most $n$ locations, thus we can find a Gray code such that $m \le n + 1$.
Since the elements ${g_i}$ and ${g_{i + 1}}$ ($1 \le i \le m - 1$) differ at only one location, we can implement the transformation $\left| {{g_i}} \right\rangle \to \left| {{g_{i + 1}}} \right\rangle $ by the ${C^n}({X_k})$ gate shown in Figure \ref{fig:x2}. For example, when $\left| {{g_2}} \right\rangle = \left| {0 \cdots 01} \right\rangle $ and $\left| {{g_3}} \right\rangle = \left| {0 \cdots 11} \right\rangle $, where $\left| {{g_2}} \right\rangle $ and $\left| {{g_3}} \right\rangle $ are two elements of the Gray code in Eq.  (\ref{eqn_318x1}), the ${C^n}({X_{n - 1}})$ gate in Figure \ref{fig:x2} sends $\left| {{g_2}} \right\rangle$ to  $\left| {{g_3}} \right\rangle $.

To understand implementation of the quantum circuits for two-point swapping more clearly, let us consider a $k$-dimensional color image
 with ${{\rm{2}}^n}$ pixels as an example. Suppose that the NASS state ${\left| \psi \right\rangle _k}$ in Eq. (\ref{eqn_316x5}) represents the image. In addition, assume that $\left| s \right\rangle = \left| 0 \right\rangle$ and $\left| t \right\rangle = \left| {{2^n} - 1} \right\rangle $ are the coordinates of the two swapped pixels. The Gray code that connects $s$ and $t$ is shown in Eq. (\ref{eqn_318x1}), where
 ${g_1},{g_2}, \ldots {g_{n + 1}}$ are the elements of the Gray code, $\left| {{g_1}} \right\rangle = \left| s \right\rangle = \left| 0 \right\rangle $, and $\left| {{g_{n + 1}}} \right\rangle = \left| t \right\rangle = \left| {{2^n} - 1} \right\rangle $. We can achieve the two-point swapping of the $k$-dimensional color image by implementing the transformations as follows (a proof is given in Appendix A).
 \begin{equation}
\label{eqn_318x2}
\left\{ \begin{array}{l}
\left| {{g_1}} \right\rangle  \to \left| {{g_2}} \right\rangle  \to  \cdots  \to \left| {{g_{n + 1}}} \right\rangle \\
\left| {{g_n}} \right\rangle  \to \left| {{g_{n - 1}}} \right\rangle  \to  \cdots  \to \left| {{g_1}} \right\rangle
\end{array} \right.
\end{equation}
 The quantum circuit for %implementation of %new
 the transformations in Eq. (\ref{eqn_318x2}) is shown in Figure \ref{fig:8x1}.
 \begin{figure}[!h]
\centering
\includegraphics[width=3.5in]{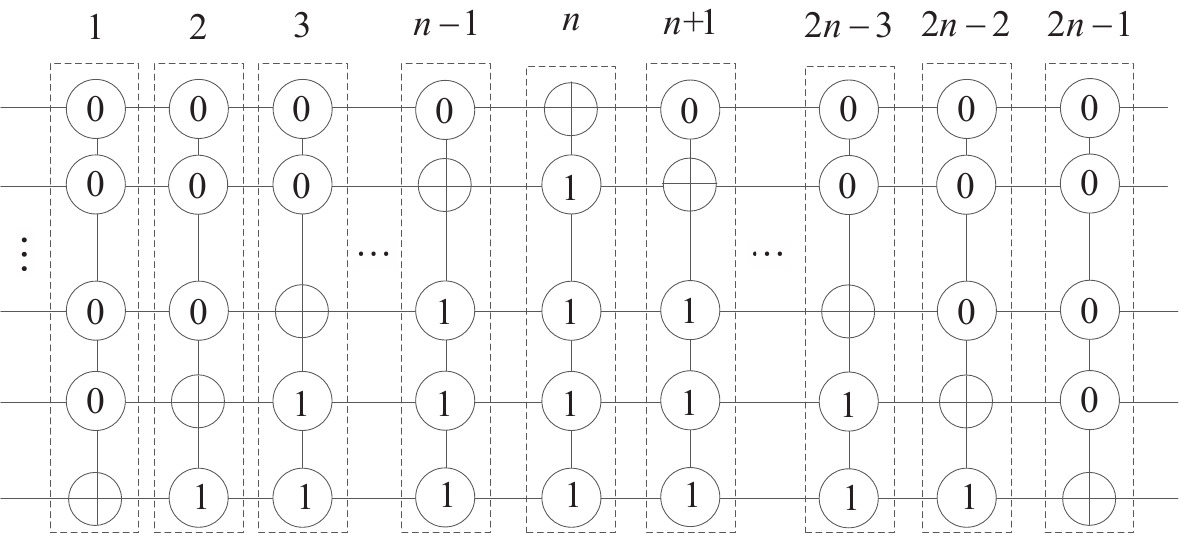}
 %where an .eps filename suffix will be assumed under latex,
 %and a .pdf suffix will be assumed for pdflatex; or what has been declared
 %via \DeclareGraphicsExtensions.
\caption{The implementation of a two-point swapping of a k-dimensional color image, where the coordinates of two points are $\left| 0 \right\rangle $ and $\left| {{2^n} - 1} \right\rangle$. The dashed box $i$ ($1 \le i \le n$) and $n+j$ ($1 \le j \le n - 1$) implement the transformations $\left| {{g_i}} \right\rangle \to \left| {{g_{i + 1}}} \right\rangle $ and $\left| {{g_{n + 1 - j}}} \right\rangle \to \left| {{g_{n - j}}} \right\rangle $, respectively.}
\label{fig:8x1}
\end{figure}

Since there are more than one Gray codes connecting $s$ and $t$, the implementation
of the two-point swapping of $s$ and $t$ can be done by
more than one quantum circuits. For example, if $s = 00101$ and $t = 11110$, two Gray codes are as follows.
\begin{equation}
\label{eqn_13}
\begin{array}{*{20}{c}}
0&0&1&0&1\\
1&0&1&0&1\\
1&1&1&0&1\\
1&1&1&1&1\\
1&1&1&1&0
\end{array}
\end{equation}
and
\begin{equation}
\label{eqn_o13}
\begin{array}{*{20}{c}}
0&0&1&0&1\\
0&0&1&0&0\\
0&0&1&1&0\\
0&1&1&1&0\\
1&1&1&1&0
\end{array}
\end{equation}

Using the Gray code in Eq. (\ref{eqn_13}), we can implement the two-point swapping of the 3D color image in Figure \ref{fig:x5}, where the coordinates of the two pixels are $\left| s \right\rangle = \left| {00} \right\rangle \left| {10} \right\rangle \left| 1 \right\rangle $ and $\left| t \right\rangle = \left| {11} \right\rangle \left| {11} \right\rangle \left| 0 \right\rangle $, i.e., $(0,2,1)$ and $(3,3,0)$. The quantum circuit and the result of the two-point swapping are shown in Figure \ref{fig:6}.
\begin{figure}[!h]
\centering
\includegraphics[width=3.5in]{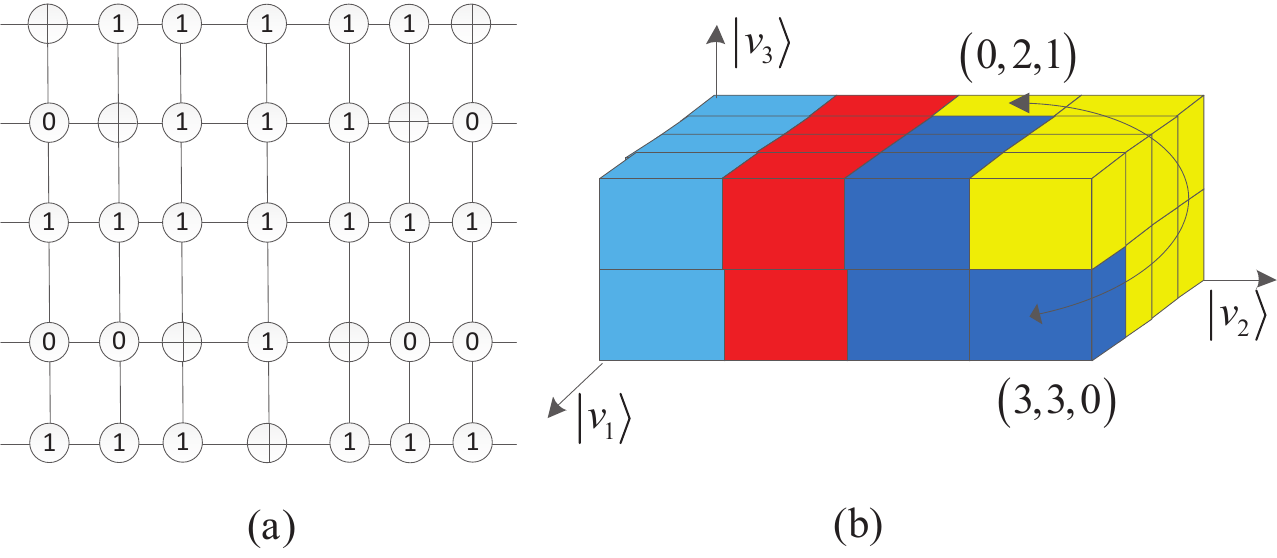}
 %where an .eps filename suffix will be assumed under latex,
 %and a .pdf suffix will be assumed for pdflatex; or what has been declared
 %via \DeclareGraphicsExtensions.
\caption{Implementation of the two-point swapping of a 3D color image. (a) The quantum circuit for two-point swapping. (b) The swapped image.}
\label{fig:6}
\end{figure}

Similarly, we employ the Gray code in Eq. (\ref{eqn_o13}) to achieve the swapping of $\left| {00101} \right\rangle $ (i.e., (0,0,1,0,1)) and $\left| {11110} \right\rangle $ (i.e., (1,1,1,1,0)) in the 5D color image in Figure \ref{fig:x6}.  The quantum circuit and the result of the two-point swapping are shown in Figure \ref{fig:x10}.
\begin{figure}[!h]
\centering
\includegraphics[width=3.5in]{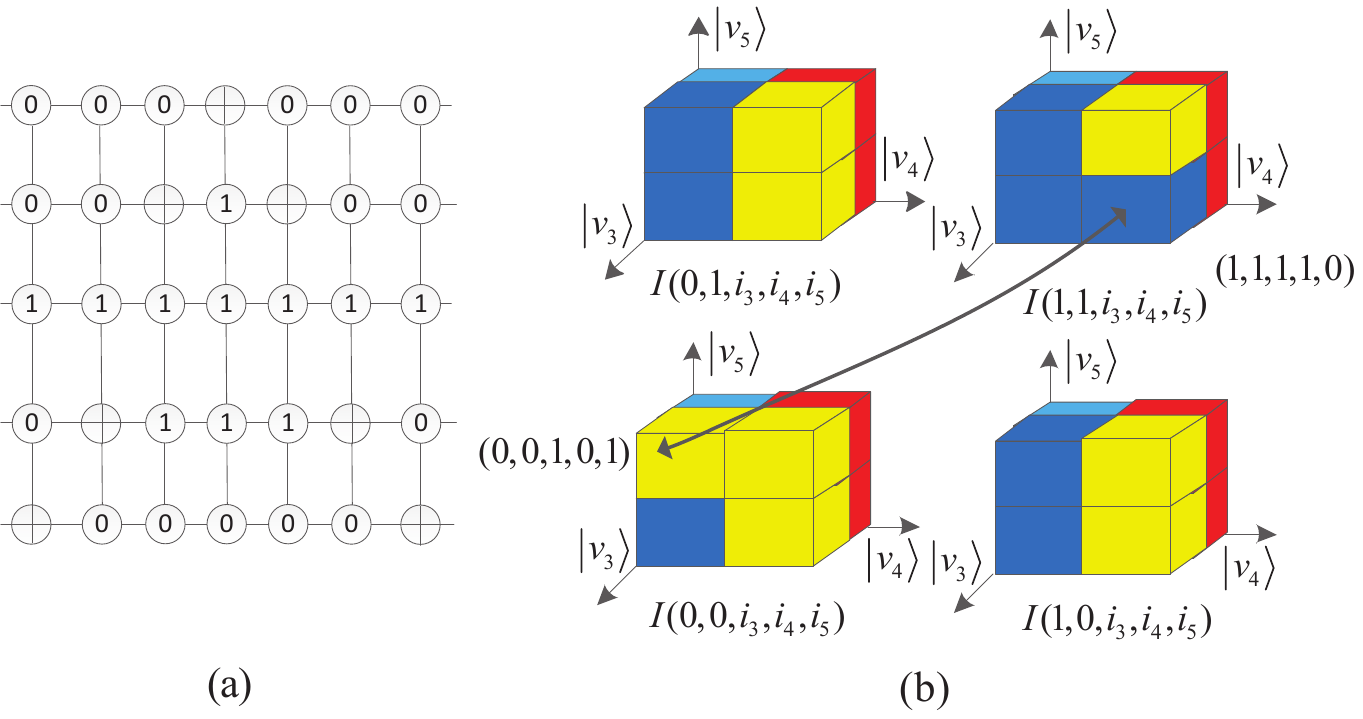}
 %where an .eps filename suffix will be assumed under latex,
 %and a .pdf suffix will be assumed for pdflatex; or what has been declared
 %via \DeclareGraphicsExtensions.
\caption{Implementation of the two-point swapping of a 5D color image. (a) The quantum circuit for two-point swapping. (b) The swapped image.}
\label{fig:x10}
\end{figure}

As any quantum circuit can be built using single-qubit and controlled-NOT gates \cite{bibitem13},
we introduce the following notion.
\begin{definition}
\label{def_8x4}
The {\it complexity} of a quantum circuit is the total number of single-qubit and controlled-NOT gates
(i.e., ${N_{C1}}$ in Figure \ref{fig:2})) in the quantum circuit.
\end{definition}

\begin{theorem}
\label{theorem1}
 Let $s={g_1},{g_2}, \cdots ,{g_m} = t$ be the $n$-bit elements of a Gray code connecting $s$ and $t$, where $\left| s \right\rangle $ and $\left| t \right\rangle $ are the coordinates of the two swapped pixels in ${G_T}$. The transformations $\left| {{{\rm{g}}_1}} \right\rangle \to \left| {{g_2}} \right\rangle \to \cdots\to \left| {{g_m}} \right\rangle $ and $\left| {{{\rm{g}}_{m - 1}}} \right\rangle \to \left| {{g_{m - 2}}} \right\rangle \to \cdots \to\left| {{g_1}} \right\rangle $ can be achieved by a sequence of ${C^n}({X_k})$ gates as shown in Figure \ref{fig:x2}. The final result is an implementation of two-point swapping operators ${G_T}$ with complexity $O({n^2})$.
\end{theorem}
\begin{proof}
See Appendix A.
\end{proof}

\subsection{\label{sec:level62} Flip transformations}

For any vector $\left| v \right\rangle = \left| {{j_1}{j_2} \cdots {j_{{m}}}} \right\rangle $,
we define
$\left| {\overline {{v}} } \right\rangle = \left| {{{\bar j}_1}{{\bar j}_2} \cdots {{\bar j}_{{m}}}} \right\rangle$, where ${\bar j_h} = 1 - {j_h}$, $h = 1,2, \cdots {m}$.
\begin{definition}
\label{def_4}
A symmetric flip $G_F^{\left| {{v_j}} \right\rangle }$ for the NASS state ${\left| \psi \right\rangle _k}$ along the $\left| {{v_j}} \right\rangle $ axis is the linear operator
\begin{equation}
\label{eqn_161}
G_F^{\left| {{v_j}} \right\rangle }({\left| \psi  \right\rangle _k}) = \sum\limits_{i = 0}^{{2^n} - 1} {{\theta _i}} \left| {\overline {{v_1}} } \right\rangle  \cdots \left| {\overline {{v_{j - 1}}} } \right\rangle \left| {v{}_j} \right\rangle \left| {\overline {{v_{j + 1}}} } \right\rangle  \cdots \left| {\overline {{v_k}} } \right\rangle
\end{equation}
where ${\left| \psi \right\rangle _k}$ is
a $k$-dimensional color image (see Eq. (\ref{eqn_3})), $\left| {{v_1}} \right\rangle , \ldots ,\left| {{v_k}} \right\rangle $ are the $k$ axes in a $k$-dimensional space.
\end{definition}

For example, by applying $G_F^{\left| {{v_j}} \right\rangle }$ to the NASS state ${\left| \psi \right\rangle _k}$ in Eq. (\ref{eqn_316x5}), the result is
\begin{equation}
\label{eqn_20X1}
\begin{array}{l}
G_F^{\left| {{v_j}} \right\rangle }({\left| \psi  \right\rangle _k})\\
 = \sum\limits_{i = 0}^{{2^n} - 1} {{\theta _i}\left| {{{\bar i}_1} \cdots {{\bar i}_{{m_1}}}} \right\rangle }  \cdots \left| {{i_{(\sum\limits_{h = 1}^{j - 1} {{m_h}} ) + 1}} \cdots {i_{(\sum\limits_{h = 1}^{j - 1} {{m_h}} ) + {m_j}}}} \right\rangle \\
 \cdots \left| {{{\bar i}_{(\sum\limits_{h = 1}^{k - 1} {{m_h}} ) + 1}} \cdots {{\bar i}_{(\sum\limits_{h = 1}^{k - 1} {{m_h}} ) + {m_k}}}} \right\rangle
\end{array}
\end{equation}
where $\dim (\left| {{v_l}} \right\rangle ) = {m_l}$, $l = 1,2, \ldots ,k$.

The symmetric flip $G_F^{\left| {{v_j}} \right\rangle }$ in Eq. (\ref{eqn_20X1}) is also expressed as
\begin{equation}
\label{eqn_20X2}
G_F^{\left| {{v_j}} \right\rangle }{\rm{ = }}{X^ \otimes }^{{m_1}} \otimes  \cdots {X^{ \otimes {m_{j - 1}}}} \otimes {I^{ \otimes {m_j}}} \otimes {X^{ \otimes {m_{j + 1}}}} \cdots  \otimes {X^{ \otimes {m_k}}}
\end{equation}
where $X$ is the Pauli spin operator (see Figure \ref{fig:1}). The implementation of the operator $G_F^{\left| {{v_j}} \right\rangle }$ in Eq. (\ref{eqn_20X2}) is shown in Figure \ref{fig:8}. Suppose that $n - \dim (\left| {{v_j}} \right\rangle ) = m$, we know that the implementation of the operator $G_F^{\left| {{v_j}} \right\rangle }$ requires $m$ Pauli-X gates, where $m<n$, i.e., the complexity of the operator $G_F^{\left| {{v_j}} \right\rangle }$ is $O(n)$.

\begin{figure}[!h]
\centering
\includegraphics[width=2in]{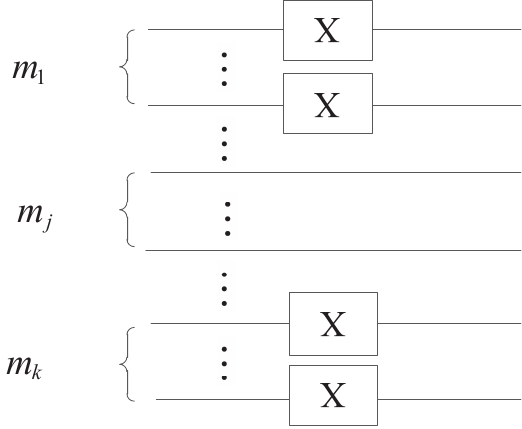}
 %where an .eps filename suffix will be assumed under latex,
 %and a .pdf suffix will be assumed for pdflatex; or what has been declared
 %via \DeclareGraphicsExtensions.
\caption{Implementation of the operator $G_F^{\left| {{v_j}} \right\rangle }$.}
\label{fig:8}
\end{figure}

Substituting $n = 5$, $k = 5$, $j = 1$, and ${m_{\rm{1}}} = {m_{\rm{2}}}{\rm{ = }}{m_{\rm{3}}}{\rm{ = }}{m_{\rm{4}}}{\rm{ = }}{m_{\rm{5}}}{\rm{ = 1}}$ into Eq. (\ref{eqn_20X1}), we obtain
\begin{equation}
\label{eqn_20X3}
\begin{array}{l}
G_F^{\left| {{v_1}} \right\rangle }({\left| \psi  \right\rangle _k}) = \sum\limits_{i = 0}^{{2^n} - 1} {{\theta _i}\left| {{i_1}} \right\rangle } \left| {{{\bar i}_2}} \right\rangle \left| {{{\bar i}_3}} \right\rangle \left| {{{\bar i}_4}} \right\rangle \left| {{{\bar i}_5}} \right\rangle \\
 = {\theta _0}\left| 0 \right\rangle \left| 1 \right\rangle \left| 1 \right\rangle \left| 1 \right\rangle \left| 1 \right\rangle  + {\theta _1}\left| 0 \right\rangle \left| 1 \right\rangle \left| 1 \right\rangle \left| 1 \right\rangle \left| 0 \right\rangle \\
 +  \cdots  + {\theta _{31}}\left| 1 \right\rangle \left| 0 \right\rangle \left| 0 \right\rangle \left| 0 \right\rangle \left| 0 \right\rangle
\end{array}
\end{equation}
which implements a symmetric flip of a 5D color image. This is shown
in part (a) of Figure \ref{fig:20x1}.
\begin{figure}[!h]
\centering
\includegraphics[width=3in]{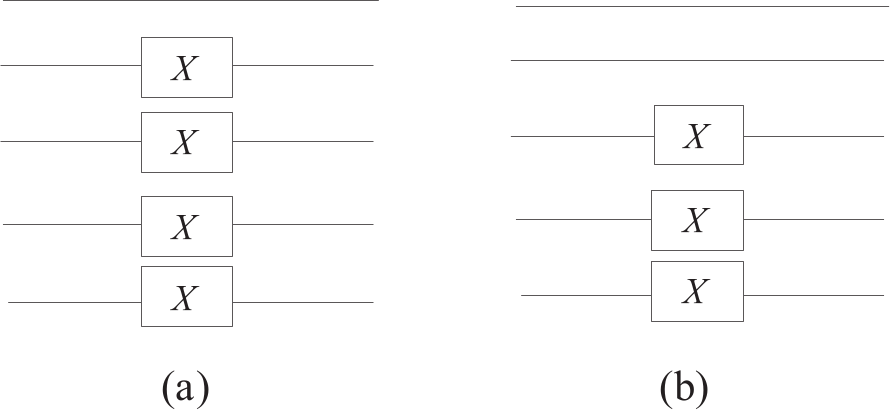}
 %where an .eps filename suffix will be assumed under latex,
 %and a .pdf suffix will be assumed for pdflatex; or what has been declared
 %via \DeclareGraphicsExtensions.
\caption{Implementation of the operator $G_F^{\left| {{v_1}} \right\rangle }$. (a) $G_F^{\left| {{v_1}} \right\rangle }$ is in Eq. (\ref{eqn_20X3}). (b) $G_F^{\left| {{v_1}} \right\rangle }$ is in Eq. (\ref{eqn_20X4}).}
\label{fig:20x1}
\end{figure}

Substituting $n = 5$, $k = 3$, $j = 1$, ${m_1} = 2$, ${m_{\rm{2}}} = 2$, and ${m_{\rm{3}}} = {\rm{1}}$ into Eq. (\ref{eqn_20X1}), we obtain
\begin{equation}
\label{eqn_20X4}
\begin{array}{l}
G_F^{\left| {{v_1}} \right\rangle }({\left| \psi  \right\rangle _3}) = \sum\limits_{i = 0}^{{2^n} - 1} {{\theta _i}\left| {{i_1}{i_2}} \right\rangle } \left| {{{\bar i}_3}{{\bar i}_4}} \right\rangle \left| {{{\bar i}_5}} \right\rangle \\
 = {\theta _0}\left| {00} \right\rangle \left| {11} \right\rangle \left| 1 \right\rangle  + {\theta _1}\left| {00} \right\rangle \left| {11} \right\rangle \left| 0 \right\rangle  +  \cdots  + {\theta _{31}}\left| {11} \right\rangle \left| {00} \right\rangle \left| 0 \right\rangle
\end{array}
\end{equation}
which realizes a symmetric flip of a 3D color image and is shown
in part (b) of Figure \ref{fig:20x1}. Application of
the quantum circuit to the image in Figure \ref{fig:x5} is shown in Figure \ref{fig:7}.
\begin{figure}[!t]
\centering
\includegraphics[width=3in]{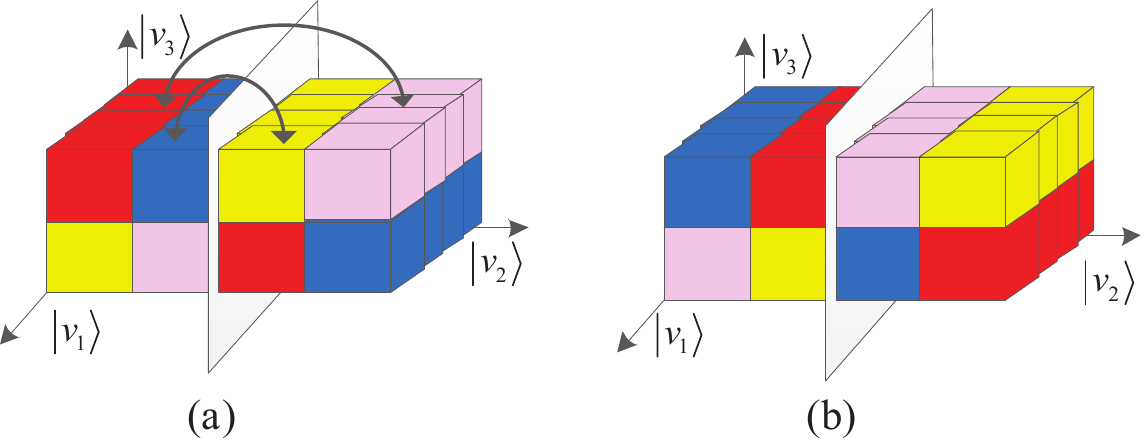}
 %where an .eps filename suffix will be assumed under latex,
 %and a .pdf suffix will be assumed for pdflatex; or what has been declared
 %via \DeclareGraphicsExtensions.
\caption{Symmetry flip of a 3D color image along the $\left| {{v_1}} \right\rangle $ axis. (a) The original image. (b) The transformed image.}
\label{fig:7}
\end{figure}

\begin{theorem}
\label{theorem2}
The operator $G_F^{\left| {{v_j}} \right\rangle }$ can be implemented by ${{\rm{2}}^{n{\rm{ - 1}}}}$ two-point swappings in Eq. (\ref{eqn_11}). The quantum circuit is reduced to the one shown in Figure \ref{fig:8}, and
the complexity of the symmetric flip is $O(n)$.
\end{theorem}
\begin{proof}
See Appendix B.
\end{proof}

\begin{definition}
\label{def_5}
A local flip $G_{LF}^{\left| {{v_x}} \right\rangle v(j,h,m)}$ for the NASS state $\left| \psi \right\rangle_{k}$ along the $\left| {{v_x}} \right\rangle $ axis is defined as the operator
\begin{equation}
\label{eqn_r20}
\begin{array}{l}
G_{LF}^{\left| {{v_x}} \right\rangle v(j,h,m)}({\left| \psi  \right\rangle _k}) = \sum\limits_{i = 0,{j_h} \ne m}^{{2^n} - 1} {{\theta _i}} \left| {{v_1}} \right\rangle  \cdots \left| {{v_k}} \right\rangle \\
 + \sum\limits_{i = 0,{j_h} = m}^{{2^n} - 1} {({\theta _i}} \left| {\overline {{v_1}} } \right\rangle  \cdots \left| {\overline {{v_{j - 1}}} } \right\rangle \left| {{{\bar j}_1} \cdots {{\bar j}_{h - 1}}{j_h}{{\bar j}_{h + 1}} \cdots {{\bar j}_{{m_j}}}} \right\rangle \\
\left| {\overline {{v_{j + 1}}} } \right\rangle  \cdots \left| {\overline {v{}_{x - 1}} } \right\rangle \left| {v{}_x} \right\rangle \left| {\overline {v{}_{x + 1}} } \right\rangle  \cdots \left| {\overline {{v_k}} } \right\rangle )
\end{array}
\end{equation}
where ${\left| \psi \right\rangle _k}$ represents a $k$-dimensional color image (see Eq. (\ref{eqn_3})),
and $\left| {{v_1}} \right\rangle , \ldots ,\left| {{v_k}} \right\rangle $ are the $k$ axes.
%in a $k$-dimensional space.
\end{definition}
Here $v(j,h,m)$ means that for the axis $\left| {{v_j}} \right\rangle = \left| {{j_1} \cdots {j_h} \cdots {j_{{m_j}}}} \right\rangle $, when ${j_h} \ne m$, the corresponding pixels are not transformed. When ${j_h} = m$, the corresponding pixels are flipped. Suppose that $\left| {{v_l}} \right\rangle = \left| {{j_1}{j_2} \cdots {j_{{m_l}}}} \right\rangle $, then $\left| {\overline {{v_l}} } \right\rangle = \left| {{{\bar j}_1}{{\bar j}_2} \cdots {{\bar j}_{{m_l}}}} \right\rangle$, ${\bar j_h} = 1 - {j_h}$, $h = 1,2, \cdots {m_l}$.

For example, application of $G_{LF}^{\left| {{v_x}} \right\rangle v(j,h,m)}$ to the NASS state ${\left| \psi \right\rangle _k}$ in Eq. (\ref{eqn_316x5}) is
\begin{equation}
\label{eqn_22X1}
\begin{array}{l}
G_{LF}^{\left| {{v_x}} \right\rangle v(j,h,m)}({\left| \psi  \right\rangle _k})\\
 = \sum\limits_{i = 0,{j_h} \ne m}^{{2^n} - 1} {{\theta _i}\left| {{i_1} \cdots {i_{{m_1}}}} \right\rangle }  \cdots \left| {{i_{{l_k} + 1}} \cdots {i_{{l_k} + {m_k}}}} \right\rangle \\
 + \sum\limits_{i = 0,{j_h} = m}^{{2^n} - 1} {{\theta _i}\left| {{{\bar i}_1} \cdots {{\bar i}_{{m_1}}}} \right\rangle }  \cdots \left| {{{\bar j}_1} \cdots {{\bar j}_{h - 1}}{j_h}{{\bar j}_{h + 1}} \cdots {{\bar j}_{{m_j}}}} \right\rangle \\
 \cdots \left| {{i_{{l_x} + 1}} \cdots {i_{{l_x} + {m_x}}}} \right\rangle  \cdots \left| {{{\bar i}_{{l_k} + 1}} \cdots {{\bar i}_{{l_k} + {m_k}}}} \right\rangle
\end{array}
\end{equation}
where ${m_y} = \dim (\left| {{v_y}} \right\rangle )$, $y = 1,2, \ldots ,k$, ${l_l} = {m_1} + {m_2} + \cdots + {m_{l - 1}}$, $l = 2, \ldots ,k$.

The implementation of $G_{LF}^{\left| {{v_x}} \right\rangle v(j,h,m)}$ is shown in Figure \ref{fig:14}. From Figure \ref{fig:14}, we know that $G_{LF}^{\left| {{v_x}} \right\rangle v(j,h,m)}$ is implemented by $n - 1 - \dim (\left| {{v_x}} \right\rangle )$ ${N_{C1}}$ or the ${N_{C0}}$ gates in Figure \ref{fig:2}, i.e., the complexity of the operator $G_{LF}^{\left| {{v_x}} \right\rangle v(j,h,m)}$ is $O(n)$.

\begin{figure}[!h]
\centering
\includegraphics[width=2.5in]{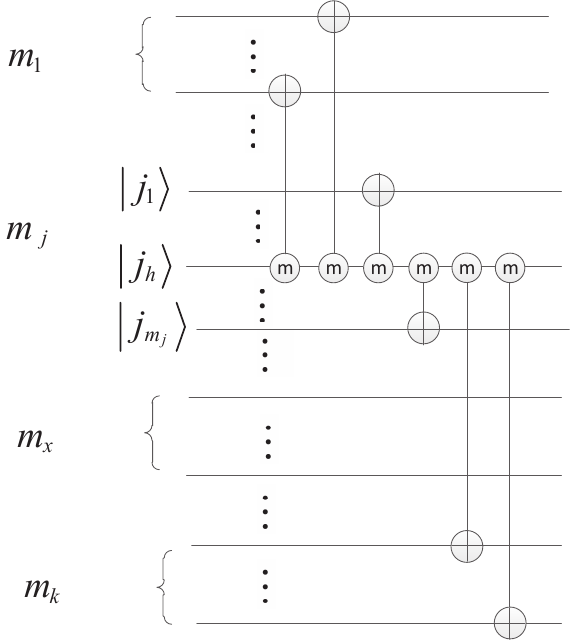}
 %where an .eps filename suffix will be assumed under latex,
 %and a .pdf suffix will be assumed for pdflatex; or what has been declared
 %via \DeclareGraphicsExtensions.
\caption{Implementation of the operator $G_{LF}^{\left| {{v_x}} \right\rangle v(j,h,m)}$ where $m=0$ or $m=1$.}
\label{fig:14}
\end{figure}

Substituting $n = 5$, $k = 5$, $x = 1$, $v(j,h,m) = v(2,1,1)$ and ${m_{\rm{1}}} = {m_{\rm{2}}}{\rm{ = }}{m_{\rm{3}}}{\rm{ = }}{m_{\rm{4}}}{\rm{ = }}{m_{\rm{5}}}{\rm{ = 1}}$ into Eq. (\ref{eqn_22X1}), we obtain
\begin{equation}
\label{eqn_22X2}
\begin{array}{l}
G_{LF}^{\left| {{v_1}} \right\rangle v(2,1,1)}({\left| \psi  \right\rangle _5})\\
 = \sum\limits_{i = 0}^{31} {{\theta _i}\left| {{i_1}} \right\rangle } \left| 0 \right\rangle \left| {{i_3}} \right\rangle \left| {{i_4}} \right\rangle \left| {{i_5}} \right\rangle  + \sum\limits_{i = 0}^{31} {{\theta _i}\left| {{i_1}} \right\rangle } \left| 1 \right\rangle \left| {{{\bar i}_3}} \right\rangle \left| {{{\bar i}_4}} \right\rangle \left| {{{\bar i}_5}} \right\rangle
\end{array}
\end{equation}
Thus, we have implemented a local flip of a 5D color image. The implementation of the operator $G_{LF}^{\left| {{v_1}} \right\rangle v(2,1,1)}$ in Eq. (\ref{eqn_22X2}) is shown in part (a) of Figure \ref{fig:22x1}.
\begin{figure}[!h]
\centering
\includegraphics[width=3in]{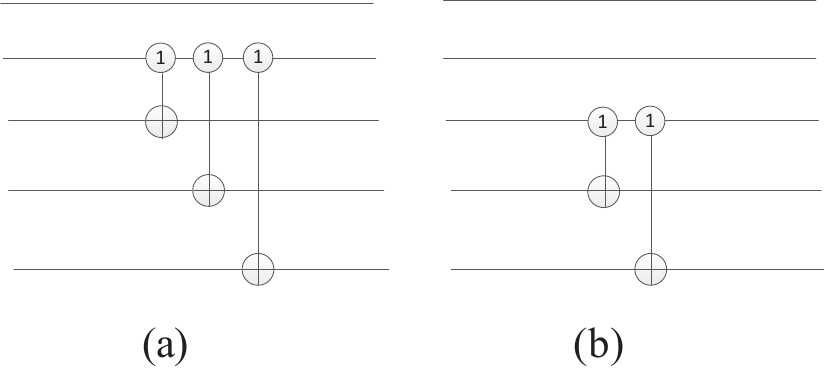}
 %where an .eps filename suffix will be assumed under latex,
 %and a .pdf suffix will be assumed for pdflatex; or what has been declared
 %via \DeclareGraphicsExtensions.
\caption{Two implementation circuits for a local flip. (a) The implementation circuit for $G_{LF}^{\left| {{v_1}} \right\rangle v(2,1,1)}$ in Eq. (\ref{eqn_22X2}). (b) The implementation circuit for $G_{LF}^{\left| {{v_1}} \right\rangle v(2,1,1)}$ in Eq. (\ref{eqn_22X3}).}
\label{fig:22x1}
\end{figure}

Substituting $n = 5$, $k = 3$, $x = 1$, ${m_1} = 2$, ${m_{\rm{2}}} = 2$, ${m_{\rm{3}}} = {\rm{1}}$, and $v(j,h,m) = v(2,1,1)$ into Eq. (\ref{eqn_22X1}), we obtain
\begin{equation}
\label{eqn_22X3}
\begin{array}{l}
G_{LF}^{\left| {{v_1}} \right\rangle v(2,1,1)}({\left| \psi  \right\rangle _3})\\
 = \sum\limits_{i = 0}^{31} {{\theta _i}\left| {{i_1}{i_2}} \right\rangle } \left| {0{i_4}} \right\rangle \left| {{i_5}} \right\rangle  + \sum\limits_{i = 0}^{31} {{\theta _i}\left| {{i_1}{i_2}} \right\rangle } \left| {1{{\bar i}_4}} \right\rangle \left| {{{\bar i}_5}} \right\rangle
\end{array}
\end{equation}
which implements a symmetric flip of a 3D color image. See part (b) of Figure \ref{fig:22x1}
for the diagram. The application of the local flip to the image in Figure \ref{fig:x5} is shown in Figure \ref{fig:9}.

\begin{figure}[!h]
\centering
\includegraphics[width=3.5in]{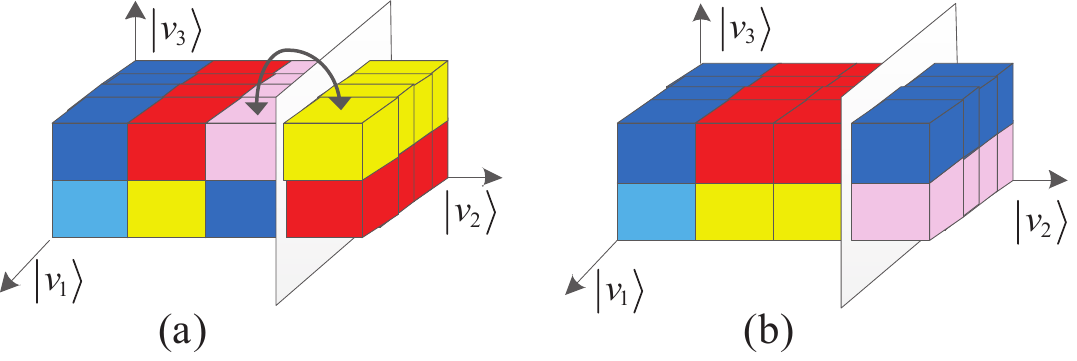}
 %where an .eps filename suffix will be assumed under latex,
 %and a .pdf suffix will be assumed for pdflatex; or what has been declared
 %via \DeclareGraphicsExtensions.
\caption{Local flip of a 3D color image. (a) The original image. (b) The transformed image.}
\label{fig:9}
\end{figure}

\subsection{\label{sec:level63} Orthogonal rotations}
\begin{definition}
\label{def_6}
An orthogonal rotation
$R_{\left| {{v_x}} \right\rangle \otimes \left| {{v_y}} \right\rangle }^\alpha $ for the NASS state ${\left| \psi \right\rangle _k}$ along the plane spanned by $\left| {{v_x}} \right\rangle \otimes \left| {{v_y}} \right\rangle $ is defined as
\begin{equation}
\label{eqn_23}
R_{\left| {{v_x}} \right\rangle  \otimes \left| {{v_y}} \right\rangle }^\alpha ({\left| \psi  \right\rangle _k}) = \sum\limits_{i = 0}^{{2^n} - 1} {{\theta _i}\left| {{v_1}} \right\rangle  \cdots \left| {v_x}' \right\rangle  \cdots \left| {v_y}' \right\rangle  \cdots \left| {{v_k}} \right\rangle }
\end{equation}
 where ${\left| \psi \right\rangle _k}$ is a $k$-dimensional color image (see Eq. (\ref{eqn_3})), $\alpha \in \left\{ {\frac{\pi }{2},\pi ,\frac{3}{2}\pi } \right\}$, $\dim (\left| {{v_x}} \right\rangle ) = \dim (\left| {{v_y}} \right\rangle )$, and
\begin{equation}
\label{eqn_24}
\left\{ {\begin{array}{*{20}{c}}
{\left| {v_x}' \right\rangle \left| {v_y}' \right\rangle  = \left| {{v_y}} \right\rangle \left| {\overline {{v_x}} } \right\rangle }&{\alpha  = \frac{\pi }{2}}\\
{\left| {v_x}' \right\rangle \left| {v_y}' \right\rangle  = \left| {\overline {{v_x}} } \right\rangle \left| {\overline {{v_y}} } \right\rangle }&{\alpha  = \pi }\\
{\left| {v_x}' \right\rangle \left| {v_y}' \right\rangle  = \left| {\overline {{v_y}} } \right\rangle \left| {{v_x}} \right\rangle }&{\alpha  = \frac{{3\pi }}{2}}
\end{array}} \right.
\end{equation}
\end{definition}

For example, the result after applying $R_{\left| {{v_x}} \right\rangle \otimes \left| {{v_y}} \right\rangle }^\alpha $ to the NASS state ${\left| \psi \right\rangle _k}$ in Eq. (\ref{eqn_316x5}) is
\begin{equation}
\label{eqn_22X4}
\begin{array}{l}
R_{\left| {{v_x}} \right\rangle  \otimes \left| {{v_y}} \right\rangle }^\alpha ({\left| \psi  \right\rangle _k})\\
 = \sum\limits_{i = 0}^{{2^n} - 1} {{\theta _i}\left| {{i_1} \cdots {i_{{m_1}}}} \right\rangle }  \cdots \left| {{{({i_{{l_x} + 1}} \cdots {i_{{l_x} + {m_x}}})}'}} \right\rangle \\
 \cdots \left| {{{({i_{{l_y} + 1}} \cdots {i_{{l_y} + {m_y}}})}'}} \right\rangle  \cdots \left| {{i_{{l_k} + 1}} \cdots {i_{{l_k} + {m_k}}}} \right\rangle
\end{array}
\end{equation}
where ${m_j} = \dim (\left| {{v_j}} \right\rangle )$, $j = 1,2, \ldots ,k$, ${m_x} = {m_y}$, ${l_l} = {m_1} + {m_2} + \cdots + {m_{l - 1}}$, $l = 2, \ldots ,k$, and
\begin{equation}
\label{eqn_22X5}
\left\{ {\begin{array}{*{20}{c}}
\begin{array}{l}
\left| {{{({i_{{l_x} + 1}} \cdots {i_{{l_x} + {m_x}}})}'}} \right\rangle \left| {{{({i_{{l_y} + 1}} \cdots {i_{{l_y} + {m_y}}})}'}} \right\rangle \\
 = \left| {{i_{{l_y} + 1}} \cdots {i_{{l_y} + {m_y}}}} \right\rangle \left| {{{\bar i}_{{l_x} + 1}} \cdots {{\bar i}_{{l_x} + {m_x}}}} \right\rangle
\end{array}&{\alpha  = \frac{\pi }{2}}\\
\begin{array}{l}
\left| {{{({i_{{l_x} + 1}} \cdots {i_{{l_x} + {m_x}}})}'}} \right\rangle \left| {{{({i_{{l_y} + 1}} \cdots {i_{{l_y} + {m_y}}})}'}} \right\rangle \\
 = \left| {{{\bar i}_{{l_x} + 1}} \cdots {{\bar i}_{{l_x} + {m_x}}}} \right\rangle \left| {{{\bar i}_{{l_y} + 1}} \cdots {{\bar i}_{{l_y} + {m_y}}}} \right\rangle
\end{array}&{\alpha  = \pi }\\
\begin{array}{l}
\left| {{{({i_{{l_x} + 1}} \cdots {i_{{l_x} + {m_x}}})}'}} \right\rangle \left| {{{({i_{{l_y} + 1}} \cdots {i_{{l_y} + {m_y}}})}'}} \right\rangle \\
 = \left| {{{\bar i}_{{l_y} + 1}} \cdots {{\bar i}_{{l_y} + {m_y}}}} \right\rangle \left| {{i_{{l_x} + 1}} \cdots {i_{{l_x} + {m_x}}}} \right\rangle
\end{array}&{\alpha  = \frac{{3\pi }}{2}}
\end{array}} \right.
\end{equation}

A swap gate is illustrated in Figure \ref{fig:10}. The implementation circuits for $R_{\left| {{v_i}} \right\rangle \otimes \left| {{v_j}} \right\rangle }^\alpha $ are shown in Figure \ref{fig:11}.
\begin{figure}[!h]
\centering
\includegraphics[width=2.5in]{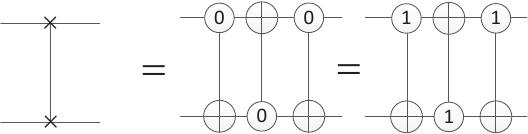}
 %where an .eps filename suffix will be assumed under latex,
 %and a .pdf suffix will be assumed for pdflatex; or what has been declared
 %via \DeclareGraphicsExtensions.
\caption{Swap gate. The swap gate (left) can be built using three ${N_{C0}}$ gates (middle) or three ${N_{C1}}$gates(right), where ${N_{C0}}$ and ${N_{C1}}$ are shown in Figure \ref{fig:2}. }
\label{fig:10}
\end{figure}
\begin{figure}[!h]
\centering
\includegraphics[width=3in]{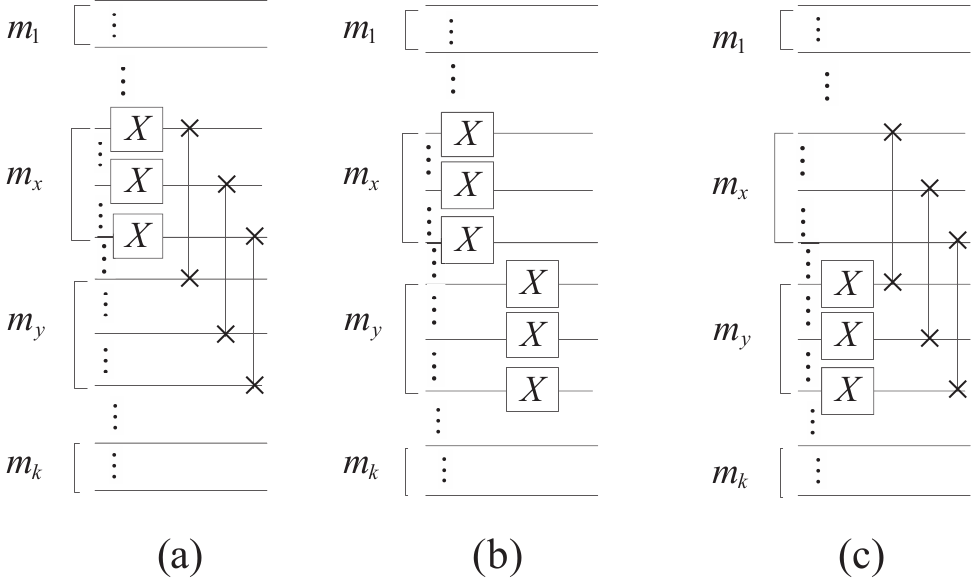}
 %where an .eps filename suffix will be assumed under latex,
 %and a .pdf suffix will be assumed for pdflatex; or what has been declared
 %via \DeclareGraphicsExtensions.
\caption{Implementation circuits for orthogonal rotation operators $R_{\left| {{v_x}} \right\rangle \otimes \left| {{v_y}} \right\rangle }^\alpha $. (a) $\alpha = \frac{\pi }{2}$. (b) $\alpha = \frac{\pi }{2}$. (c) $\alpha = \pi $ and $\alpha = \frac{{3\pi }}{2}$.}
\label{fig:11}
\end{figure}

It follows from Figure \ref{fig:11} that $R_{\left| {{v_x}} \right\rangle \otimes \left| {{v_y}} \right\rangle }^\alpha $ can be built from ${m_x}$ swap gates shown in Figure \ref{fig:10} and ${m_x}$ $X$ gates shown in Figure \ref{fig:1} for $\alpha = \frac{\pi }{2}$ or $\alpha = \frac{{3\pi }}{2}$. When $\alpha = \pi$, $R_{\left| {{v_x}} \right\rangle \otimes \left| {{v_y}} \right\rangle }^\alpha $ can be built by $2{m_x}$ $X$ gates. Because ${m_x} = \dim (\left| {{v_x}} \right\rangle ) < n$, the complexity of the orthogonal rotation $R_{\left| {{v_x}} \right\rangle \otimes \left| {{v_y}} \right\rangle }^\alpha $ is $O(n)$.

\subsection{\label{sec:level64} Translations}
\begin{definition}
\label{def_7}
A translation ${T_{\left| {{v_x}} \right\rangle }}$ for the NASS state $\left| \psi \right\rangle_{k} $ along the $\left| {{v_x}} \right\rangle $ axis is defined as the operator
\begin{equation}
\label{eqn_22X7}
{T_{\left| {{v_x}} \right\rangle }}({\left| \psi  \right\rangle _k}) = \sum\limits_{i = 0}^{{2^n} - 1} {{\theta _i}} \left| {{v_1}} \right\rangle  \cdots \left| {{v_{x - 1}}} \right\rangle \left| {v_x}' \right\rangle \left| {{v_{x + 1}}} \right\rangle  \cdots \left| {{v_k}} \right\rangle
\end{equation}
where ${\left| \psi \right\rangle _k}$ represents a $k$-dimensional color image (see Eq. (\ref{eqn_3})), $\dim (\left| {{v_x}} \right\rangle ) = {m_x}$. Suppose that $\left| {{v_x}} \right\rangle = \left| {{x_1}{x_2} \cdots {x_{{m_x}}}} \right\rangle = \left| j \right\rangle $, where ${j = {x_1}{x_2} \cdots {x_{{m_x}}}}$ is the binary expansion of the integer $j$, then
\begin{equation}
\label{eqn_22X8}
\left\{ \begin{array}{l}
\begin{array}{*{20}{c}}
{\left| {v_x}' \right\rangle  = \left| {j + 1} \right\rangle }&{0 \le j = {x_1}{x_2} \cdots {x_{{m_x}}} \le {2^{{m_x}}} - 2}
\end{array}\\
\begin{array}{*{20}{c}}
{\left| {v_x}' \right\rangle  = \left| 0 \right\rangle }&{{\rm{      }}j = {2^{{m_x}}} - 1}
\end{array}
\end{array} \right.
\end{equation}
\end{definition}

For example, when $\dim (\left| {{v_i}} \right\rangle ) = {m_i}$, $i = 1,2, \ldots ,k$, the translation operator ${T_{\left| {{v_x}} \right\rangle }}$ of a $k$-dimensional color image is also expressed as
\begin{equation}
\label{eqn_28}
\begin{array}{l}
{T_{\left| {{v_x}} \right\rangle }} = {I^{ \otimes {m_1}}} \otimes  \cdots {I^{ \otimes {m_{x - 1}}}}\\
 \otimes \left( {\left( {\sum\limits_{j = 0}^{{2^{{m_x}}} - 2} {\left| {j + 1} \right\rangle \left\langle j \right|} } \right) + \left| 0 \right\rangle \left\langle {{2^m} - 1} \right|} \right)\\
 \otimes {I^{ \otimes {m_{x + 1}}}} \cdots  \otimes {I^{ \otimes {m_k}}}
\end{array}
\end{equation}

The inside factor is called the key of implementing the translation ${T_{\left| {{v_x}} \right\rangle }}$ in Eq. (\ref{eqn_28}):
 \begin{equation}
\label{eqn_22X10}
{T_{key}} = \left( {\sum\limits_{j = 0}^{{2^{{m_x}}} - 2} {\left| {j + 1} \right\rangle \left\langle j \right|} } \right) + \left| 0 \right\rangle \left\langle {{2^m} - 1} \right|
\end{equation}
and the operator ${T_{key}}$ in Eq. (\ref{eqn_22X10}) can be implemented by two-point swappings successively as follows.
\begin{equation}
\label{eqn_22X11}
\begin{array}{l}
\left| {{2^{{m_x}}} - 1} \right\rangle  \leftrightarrow \left| 0 \right\rangle \\
\left| {{2^{{m_x}}} - 2} \right\rangle  \leftrightarrow \left| {{2^{{m_x}}} - 1} \right\rangle \\
\left| {{2^{{m_x}}} - 3} \right\rangle  \leftrightarrow \left| {{2^{{m_x}}} - 2} \right\rangle \\
 \vdots \\
\left| 1 \right\rangle  \leftrightarrow \left| 2 \right\rangle
\end{array}
\end{equation}
%where the proof of the theorem is shown in Theorem \ref{theorem4}.

\begin{theorem}
\label{theorem4}
The operator ${T_{key}}$ in Eq. (\ref{eqn_22X10}) can be implemented by ${{\rm{2}}^{\rm{m}}}{\rm{ - 1}}$ two-point swapping operators. The complexity of the operator ${T_{key}}$ is $O({2^m}{m^2})$, and the complexity of the translation ${T_{\left| {{v_x}} \right\rangle }}$ is also $O({2^m}{m^2})$,
where $\dim (\left| {{v_x}} \right\rangle ) = m$.
\end{theorem}
\begin{proof}
See Appendix C.
\end{proof}

The quantum circuit for the translation $T_{\left| {{v_x}} \right\rangle }^{}$ is shown in Figure \ref{fig:t12}.

 \begin{figure}[!h]
\centering
\includegraphics[width=2in]{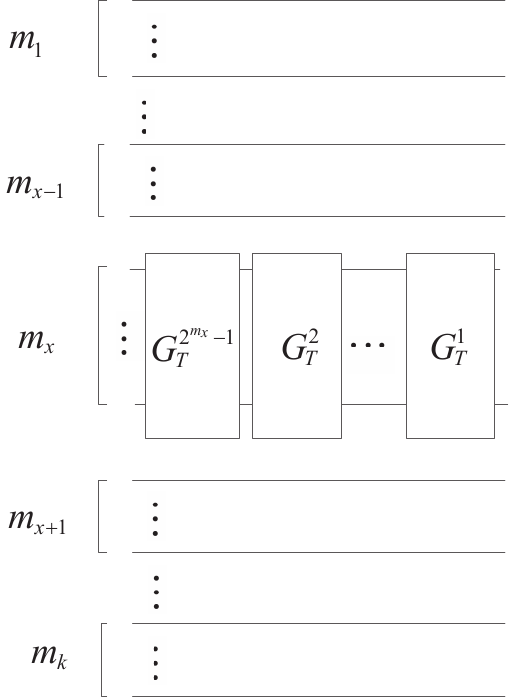}
 %where an .eps filename suffix will be assumed under latex,
 %and a .pdf suffix will be assumed for pdflatex; or what has been declared
 %via \DeclareGraphicsExtensions.
\caption{Implementation circuit for the translation ${T_{\left| {{v_x}} \right\rangle }}$, where $G_T^{{2^{{m_x}}} - 1},G_T^{{2^{{m_x}}} - {\rm{2}}}, \cdots G_T^{\rm{1}}$ are the circuits
to implement two-point swapping successively in Eq. (\ref{eqn_22X11}) using Gray codes.}
\label{fig:t12}
\end{figure}

Substituting $n = 5$, $k = 5$, $x = 4$ and ${m_{\rm{1}}} = {m_{\rm{2}}}{\rm{ = }}{m_{\rm{3}}}{\rm{ = }}{m_{\rm{4}}}{\rm{ = }}{m_{\rm{5}}}{\rm{ = 1}}$ into Eq. (\ref{eqn_28}), we get the translation
 ${T_{\left| {{v_4}} \right\rangle }}$ for a 5D color image along $\left| {{v_4}} \right\rangle $ axis
\begin{equation}
\label{eqn_22X14}
{T_{\left| {{v_4}} \right\rangle }} = I \otimes I \otimes I \otimes \left( {\left| 1 \right\rangle \left\langle 0 \right| + \left| 0 \right\rangle \left\langle 1 \right|} \right) \otimes I
\end{equation}

After applying the operator ${T_{\left| {{v_4}} \right\rangle }}$ in Eq. (\ref{eqn_22X14}) to the image in Figure \ref{fig:x6}, the quantum circuit and the result obtained are shown in Figure \ref{fig:20}.

\begin{figure}[!h]
\centering
\includegraphics[width=3.5in]{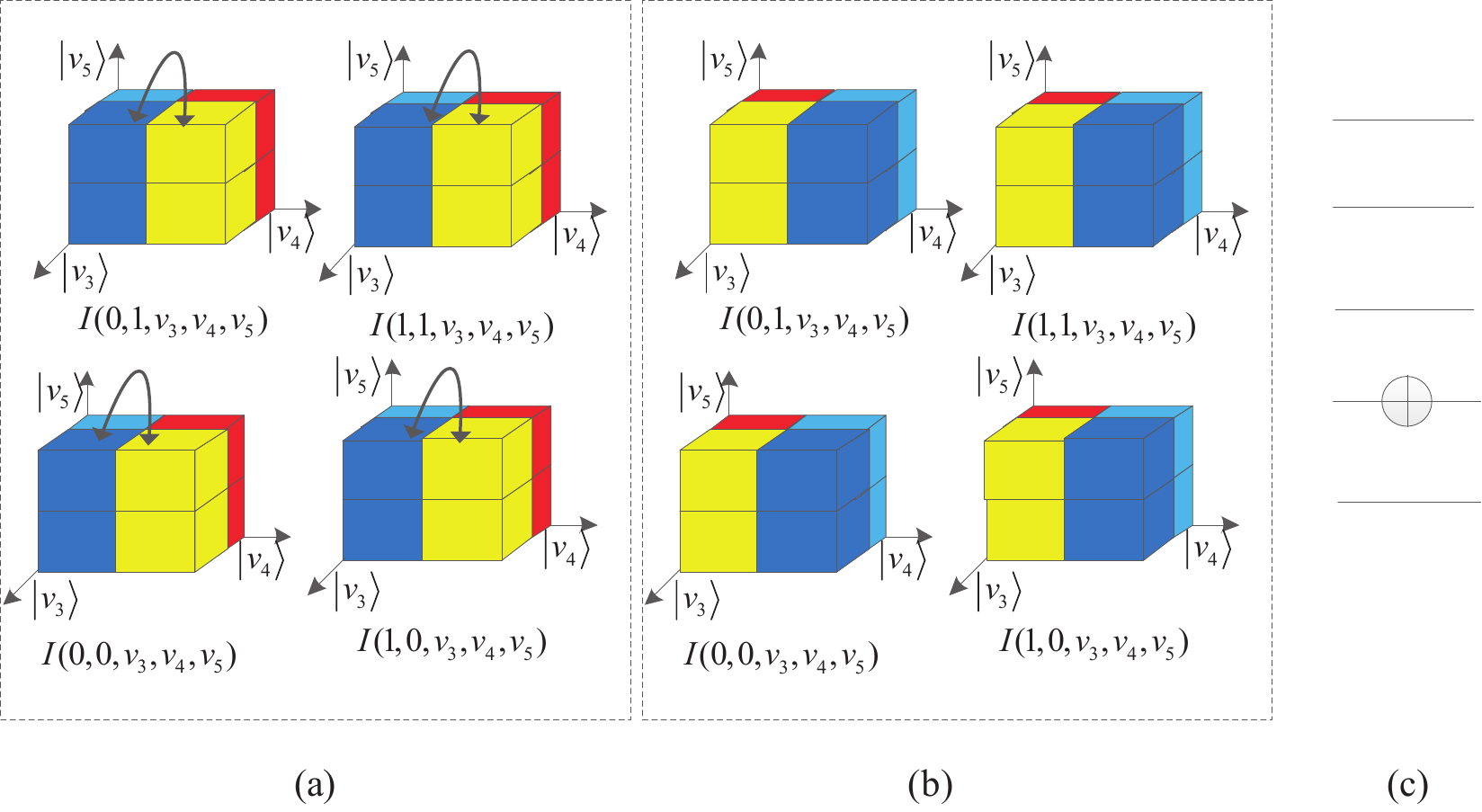}
 %where an .eps filename suffix will be assumed under latex,
 %and a .pdf suffix will be assumed for pdflatex; or what has been declared
 %via \DeclareGraphicsExtensions.
\caption{Translation of a 5D color image along $\left| {{v_4}} \right\rangle $ axis. (a) The original image. (b) The transformed image. (c) The implementation circuit. }
\label{fig:20}
\end{figure}

Substituting $n = 5$, $k = 3$, $x = 2$, ${m_1} = 2$, ${m_{\rm{2}}} = 2$, and ${m_{\rm{3}}} = {\rm{1}}$ into Eq. (\ref{eqn_28}), we obtain
\begin{equation}
\label{eqn_22X15}
{T_{\left| {{v_2}} \right\rangle }} = {I^{ \otimes 2}} \otimes \left( {\sum\limits_{j = 0}^2 {\left| {j + 1} \right\rangle \left\langle j \right|}  + \left| 0 \right\rangle \left\langle 3 \right|} \right) \otimes I
\end{equation}
which implements a translation of a 3D color image and is shown in Figure \ref{fig:13}.
The circuits in the three dashed boxes, $G_T^{\rm{3}}$, $G_T^{\rm{2}}$, and $G_T^{\rm{1}}$, implement successively %new
two-point swappings: $\left| {\rm{3}} \right\rangle \leftrightarrow \left| {\rm{0}} \right\rangle $, $\left| {\rm{2}} \right\rangle \leftrightarrow \left| {\rm{3}} \right\rangle $ and $\left| {\rm{1}} \right\rangle \leftrightarrow \left| {\rm{2}} \right\rangle $. The application of the quantum circuit in part (a) of Figure \ref{fig:13} to the image in Figure \ref{fig:x5} is shown in Figure \ref{fig:12}.

\begin{figure}[!h]
\centering
\includegraphics[width=3in]{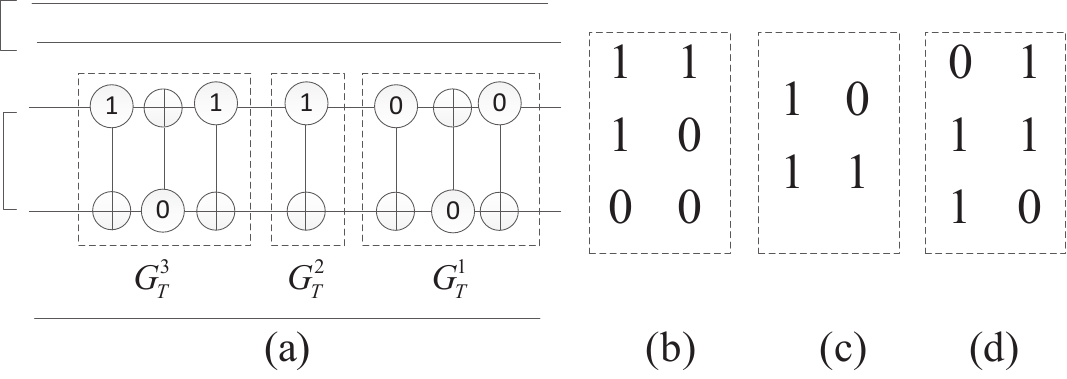}
 %where an .eps filename suffix will be assumed under latex,
 %and a .pdf suffix will be assumed for pdflatex; or what has been declared
 %via \DeclareGraphicsExtensions.
\caption{Implementation circuit for a 3D color image. (a) The quantum circuit. (b) Gray code for $G_T^{\rm{3}}$. (c) Gray code for $G_T^{\rm{2}}$. (d) Gray code for $G_T^{\rm{1}}$.}
\label{fig:13}
\end{figure}

\begin{figure}[!h]
\centering
\includegraphics[width=3in]{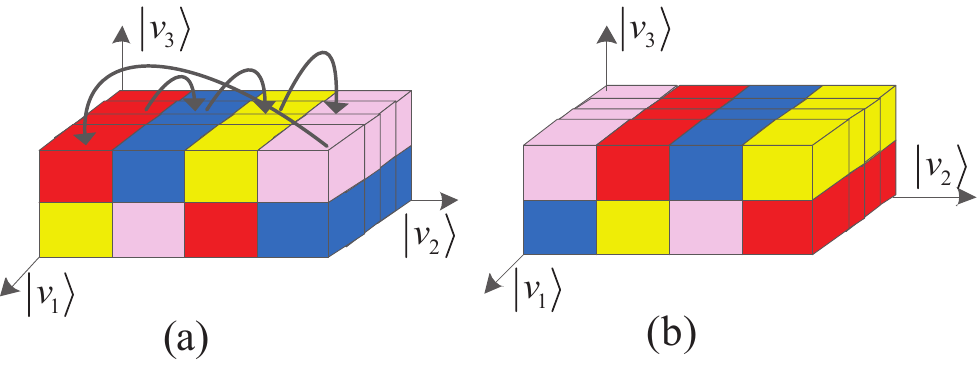}
 %where an .eps filename suffix will be assumed under latex,
 %and a .pdf suffix will be assumed for pdflatex; or what has been declared
 %via \DeclareGraphicsExtensions.
\caption{Translation for a 3D color image along  $\left| {{v_2}} \right\rangle $ axis. (a) The original image. (b) The transformed image.}
\label{fig:12}
\end{figure}

\subsection{\label{sec:level65} Comparison of geometric transformations}

Realization of images by geometric transformations over
a quantum system is a relatively new research topic. To the best of our knowledge, %new
geometric transformations based on FRQI (GTBFRQI) \cite{bibitem8,bibitem19} seems to be the only available method.
Compared with GTBFRQI, our method is capable of handling color and higher dimensional images,
while GTBFRQI is only suitable for 2D gray images.
Table \ref{table_1} and Table \ref{table_2} show the comparison of the two protocols, where $N$ is the number of  pixels in the image. Table \ref{table_1} shows the result of our method
to multidimensional color and gray images.
To store an image with $N$ pixels, our proposed method requires $\log N$ qubits, whereas  GTBFRQI needs $\log N+1$  qubits.  Table \ref{table_2} also highlights that
our method includes new transformation functions such as local flips and translations, which were absent in GTBFRQI.

On classical computers, global operators of geometric transformations for an $\sqrt N  \times \sqrt N $ image
are done by $\sqrt N  \times \sqrt N $-matrices, so the complexity of implementation
is at least $O\left( N \right)$ \cite{bibitema54,bibitem17}.
In our proposed quantum system, the global operators (symmetric flips, local flips,
and orthogonal rotations) in Table \ref{table_2}
can be implemented by $O\left( {\log N} \right)$ gates using quantum parallel computing.

\begin{table}[!h]
%% increase table row spacing, adjust to taste
\scriptsize
\renewcommand{\arraystretch}{1.3}
\caption{Image types for geometric transformations.}
\label{table_1}
\centering
%% Some packages, such as MDW tools, offer better commands for making tables
%% than the plain LaTeX2e tabular which is used here.
\begin{tabular}{c|c|c}

\hline
 &	GTBNASS &	GTBFRQI\\
 \hline
 Image types & Multidimensional color and gray image &	2D gray image\\
  \hline
 Qubit number of image storage & $\log N$ &	$\log N + 1$\\
\hline
\end{tabular}
\end{table}

\begin{table}[!h]
%% increase table row spacing, adjust to taste
\scriptsize
\renewcommand{\arraystretch}{1.3}
\caption{Complexity of geometric transformations.}
\label{table_2}
\centering
%% Some packages, such as MDW tools, offer better commands for making tables
%% than the plain LaTeX2e tabular which is used here.
\begin{tabular}{c|c|c}

\hline
   &	GTBNASS &	GTBFRQI\\
 \hline
 Two-point swapping &  $O\left( {2\log N} \right)$ &	$O\left( {2\log N + 2} \right)$\\
 \hline
 Symmetric flip & $O\left( {\log N} \right)$ &	$O\left( {\log N + 1} \right)$\\
  \hline
Local flip & $O\left( {\log N} \right)$ &	$ - $\\
  \hline
Orthogonal rotation & $O\left( {\log N} \right)$ &	$O\left( {\log N + 1} \right)$\\
   \hline
Translation & $ < O\left( {N\log N} \right)$ &	$ - $\\
\hline
\end{tabular}
\end{table}

\section{\label{sec:level13} Experiments of geometric transformations }
In the absence of the quantum computer to implement proposed geometric transformations, experiments of quantum images are simulated on a classical computer. The quantum images are stored in column vectors and geometric transformations are implemented using unitary matrices in Matlab.

A $128 \times 128$ image is regarded as the input image shown in Figure \ref{fig:24}. Results of symmetric flips
are showed in (a) and (b) of Figure \ref{fig:25}. The operators of (a) and (b) of Figure \ref{fig:25} are $G_F^{\left| X \right\rangle }{\rm{ = }}{X^{ \otimes 7}} \otimes {I^{ \otimes 7}}$ and $G_F^{\left| Y \right\rangle }{\rm{ = }}{I^{ \otimes 7}} \otimes {X^{ \otimes 7}}$.

\begin{figure}[!h]
\centering
\includegraphics[width=1.5in]{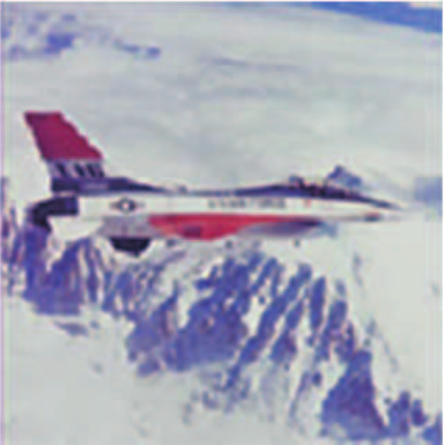}
 %where an .eps filename suffix will be assumed under latex,
 %and a .pdf suffix will be assumed for pdflatex; or what has been declared
 %via \DeclareGraphicsExtensions.
\caption{Original image}
\label{fig:24}
\end{figure}

\begin{figure}[!h]
\centering
\includegraphics[width=3in]{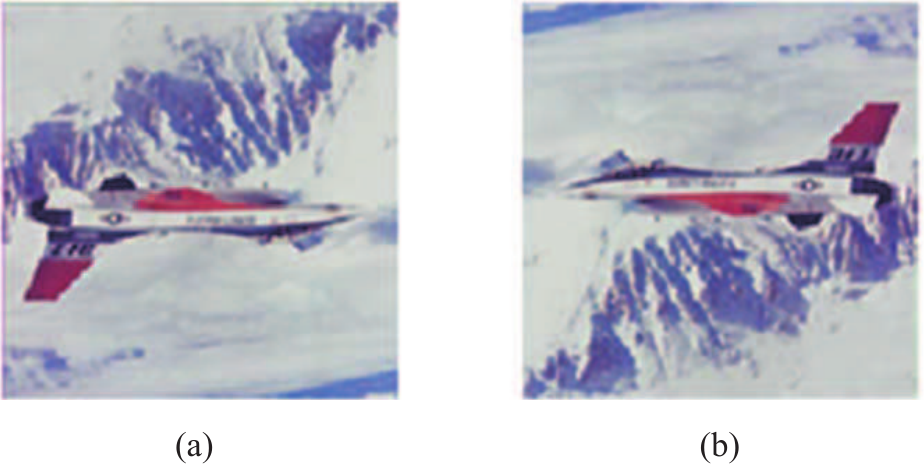}
 %where an .eps filename suffix will be assumed under latex,
 %and a .pdf suffix will be assumed for pdflatex; or what has been declared
 %via \DeclareGraphicsExtensions.
\caption{Experiments of symmetric flips. (a) Flip along X axis. (b) Flip along $Y$ axis.}
\label{fig:25}
\end{figure}

$G_{LF}^{\left| {{v_1}} \right\rangle v(2,1,1)}$ and $G_{LF}^{\left| {{v_1}} \right\rangle v(2,1,0)}$ are the operators to flip the lower part and the upper part of the original image shown in Figure \ref{fig:24}, respectively. Applications of  $G_{LF}^{\left| {{v_1}} \right\rangle v(2,1,1)}$ and $G_{LF}^{\left| {{v_1}} \right\rangle v(2,1,0)}$ to the original image of Figure \ref{fig:24} are showed in (a) and (b) of Figure \ref{fig:26}.

\begin{figure}[!h]
\centering
\includegraphics[width=3in]{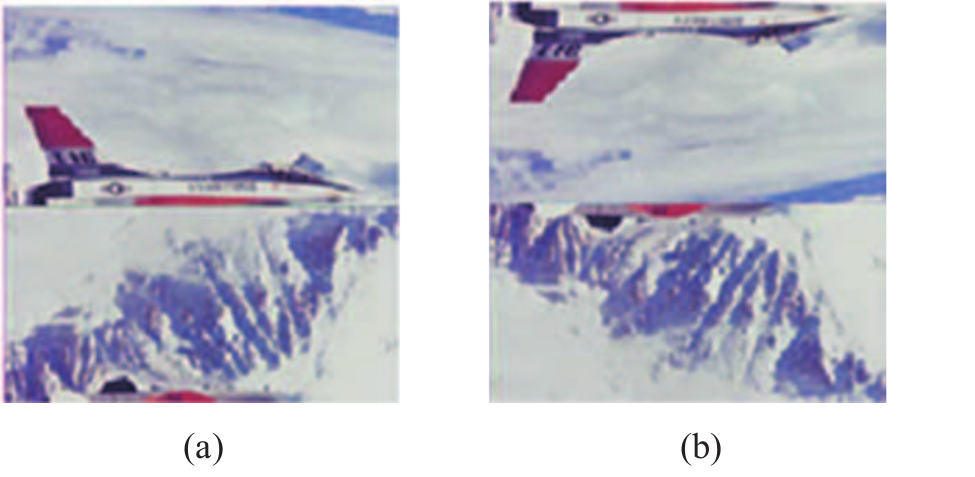}
 %where an .eps filename suffix will be assumed under latex,
 %and a .pdf suffix will be assumed for pdflatex; or what has been declared
 %via \DeclareGraphicsExtensions.
\caption{Experiments of local flips. (a) Local flip for the lower part. (b)  Local flip for the upper part.}
\label{fig:26}
\end{figure}

Application of the orthogonal rotations $R_{\left| {{v_x}} \right\rangle  \otimes \left| {{v_y}} \right\rangle }^{{\pi  \mathord{\left/
 {\vphantom {\pi  2}} \right.
 \kern-\nulldelimiterspace} 2}}$, $R_{\left| {{v_x}} \right\rangle  \otimes \left| {{v_y}} \right\rangle }^\pi $  and  $R_{\left| {{v_x}} \right\rangle  \otimes \left| {{v_y}} \right\rangle }^{{{3\pi } \mathord{\left/
 {\vphantom {{3\pi } 2}} \right.
 \kern-\nulldelimiterspace} 2}}$ on the image of Figure \ref{fig:24} are showed in Figure \ref{fig:29}.

\begin{figure}[!h]
\centering
\includegraphics[width=5.5in]{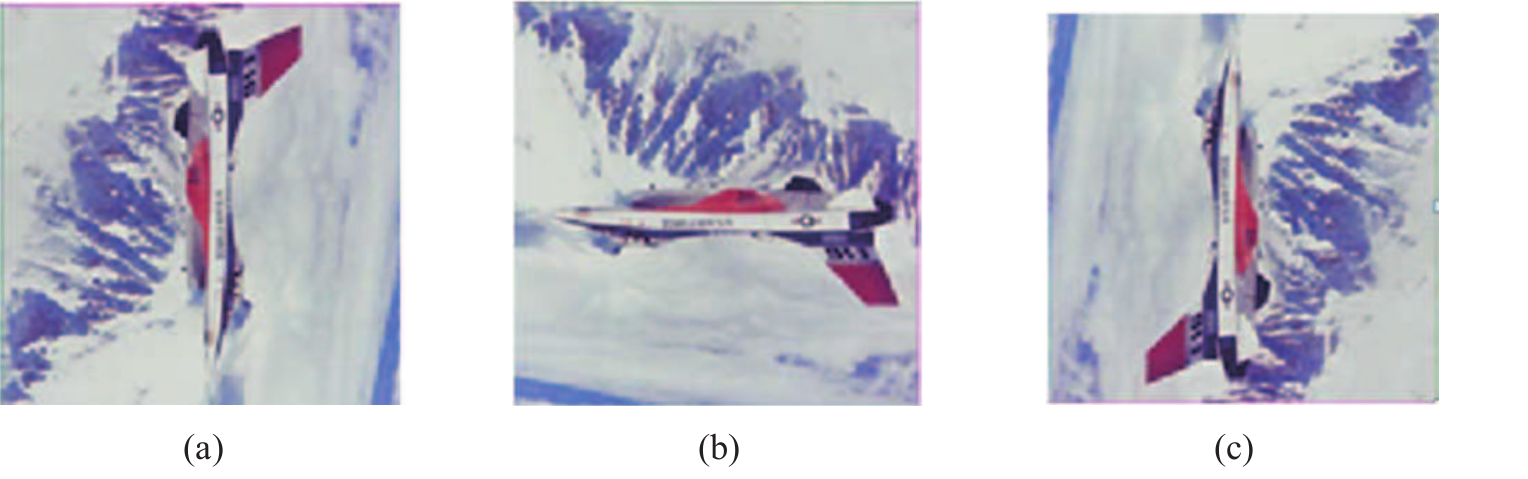}
 %where an .eps filename suffix will be assumed under latex,
 %and a .pdf suffix will be assumed for pdflatex; or what has been declared
 %via \DeclareGraphicsExtensions.
\caption{Experiments of orthogonal rotations. (a) Rotation by $\frac{\pi }{{\rm{2}}}$. (b) Rotation by $\pi $. (c) Rotation by $\frac{{3\pi }}{2}$. }
\label{fig:29}
\end{figure}

The translation operator ${T_{\left| {{v_y}} \right\rangle }} = \left( {\left( {\sum\limits_{j = 0}^{{2^7} - 2} {\left| {j + 1} \right\rangle \left\langle j \right|} } \right) + \left| 0 \right\rangle \left\langle {{2^7} - 1} \right|} \right) \otimes {I^{ \otimes 7}}$  is applied repeatedly 10 times and 30 times on the image of Figure \ref{fig:24}. Experimental data are showed in Figure \ref{fig:31}.

\begin{figure}[!h]
\centering
\includegraphics[width=3in]{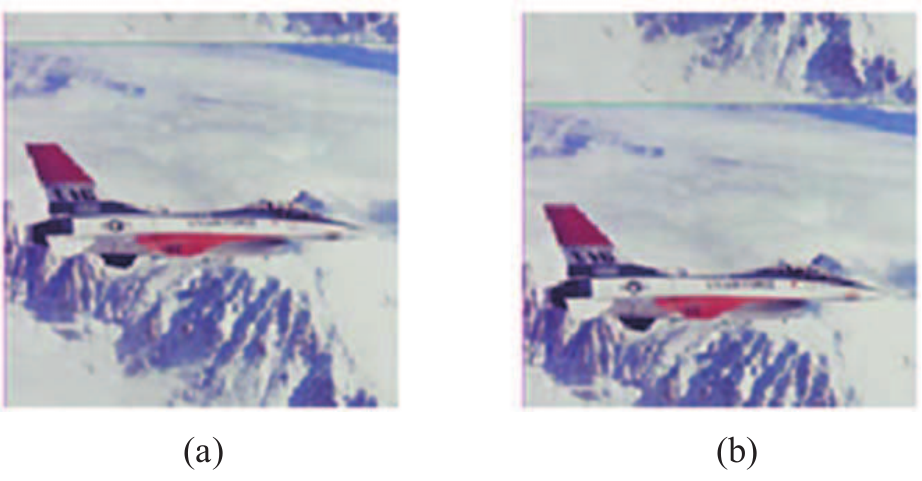}
 %where an .eps filename suffix will be assumed under latex,
 %and a .pdf suffix will be assumed for pdflatex; or what has been declared
 %via \DeclareGraphicsExtensions.
\caption{Experiments of translations. (a) Translation of 10 pixels along $Y$ axis. (b)  Translation of 30 pixels along $Y$ axis.}
\label{fig:31}
\end{figure}

\section{\label{sec:level14} Conclusions} %new
In this study, we have shown how an n-qubit NASS state represents a k-dimensional color image of $N = {2^n}$ pixels,
which demonstrates the vast capacity of NASS to store information. For a multidimensional color image based on NASS, the quantum circuits described in this study provide a convenient and efficient method for implementing geometric transformations of two-point swappings, symmetric flips, local flips, orthogonal rotations, and translations.
The complexity analysis of the geometric transformations has shown that the global operators (symmetric flips, local flips, orthogonal rotations) can be implemented by $O(n)$ gates, and the local operators (two point swappings) by $O({n^2})$ gates. This proves that the global operators are faster than the local ones, while it is
the opposite %in the local operations are faster
in classical image processing. The translation of a k-dimensional color image is achieved without additional storage space due to quantum parallel computing, which is
difficult to achieve on a classical computer.
The complexity of all geometric transformations (both global and local)
are lower than $O(N\log N)$, where $N = {2^n}$ is the number of pixels, which implies that
a geometric transformation can be constructed with a polynomial number of single-qubit and two-qubit gates. The results of simulated experiments further confirm the validity of proposed geometric transformations.
The proposed quantum scheme is one step towards more applications on quantum images with relative low complexity.
It is hoped that similar common transformations such as affine and non-rigid transformations can be implemented.
\appendix

\section{Proof of Theorem \ref{theorem1}}
\begin{proof}
\label{proof1}
Note that ${{\rm{g}}_i}$ and ${{\rm{g}}_{i + 1}}$ ($1 \le i \le m - 1$) differ by exactly one bit, say
at the $j$th bit, i.e.,
\begin{equation}
\label{eqn_19x1}
\left\{ \begin{array}{l}
{g_i} = {i_1} \cdots {i_j} \cdots {i_n}\\
{g_{i + 1}} = {i_1} \cdots {\overline i _j} \cdots {i_n}
\end{array} \right.
\end{equation}
 where ${\overline i _j} = 1 - {i_j}$, $1 \le j \le n$.

Note that the quantum circuit $X$ sends $|i_j\rangle$ to $|\bar{i_j}\rangle$, therefore
% \begin{equation}
%\label{eqn_y141}
%{C_i}\left| {{g_i}} \right\rangle  = \left| {{g_{i + 1}}} \right\rangle
%\end{equation}
the ${C^n}({X_j})$ gate in Figure \ref{fig:x2} implements that
 \begin{align}
\label{eqn_19x2}
{C^n}({X_j})\left| {{g_i}} \right\rangle  &= \left| {{g_{i + 1}}} \right\rangle\\ \label{eqn_19x3}
{C^n}({X_j})\left| {{g_{i + 1}}} \right\rangle &= \left| {{g_i}} \right\rangle
\end{align}
and for $\left| x \right\rangle \ne \left| {{g_i}} \right\rangle, \left| {{g_{i + 1}}} \right\rangle$,
\begin{equation}
\label{eqn_19x5}
{C^n}({X_j})\left| x \right\rangle  = \left| x \right\rangle
\end{equation}
Denote this ${C^n}({X_j})$ by $C_i$, i.e. $C_i|g_i\rangle=|g_{i+1}\rangle$,
$C_i|g_{i+1}\rangle=|g_{i}\rangle$, and $C_i$ fixes other $|g_j\rangle's$.

Let ${C_T} = {C_1}{C_2}\cdots {C_{m - 2}}{C_{m - 1}}{C_{m - 2}}\cdots {C_1}$. From (\ref{eqn_19x2}), (\ref{eqn_19x3}), and (\ref{eqn_19x5}) it follows that ${C_T}$ is built by $2m - 3$ gates ${C^n}({X_y})$ ($1 \le y \le n$), which implements the transformations $\left| {{g_1}} \right\rangle \to \left| {{g_2}} \right\rangle \to \cdots\to \left| {{g_m}} \right\rangle $ and $\left| {{g_{m - 1}}} \right\rangle \to \left| {{g_{m - 2}}} \right\rangle \to \cdots \to\left| {{g_1}} \right\rangle $, and we have
\begin{equation}
\label{eqn_19x6}
\left\{ \begin{array}{l}
{C_T}\left| s \right\rangle  = {C_T}\left| {{g_1}} \right\rangle  = \left| {{g_m}} \right\rangle  = \left| t \right\rangle \\
{C_T}\left| t \right\rangle  = {C_T}\left| {{g_m}} \right\rangle  = \left| {{g_1}} \right\rangle  = \left| s \right\rangle \\
{C_T}\left| x \right\rangle  = \left| x \right\rangle
\end{array} \right.
\end{equation}
where $\left| x \right\rangle \ne \left| s \right\rangle $ and $\left| x \right\rangle \ne \left| t \right\rangle $. Thus, by applying ${C_T}$ to the NASS state $\left| \psi \right\rangle_{k}$ in Eq. (\ref{eqn_3}), we have
\begin{equation}
\label{eqn_19x7}
\begin{array}{l}
{C_T}({\left| \psi  \right\rangle _k}) = \sum\limits_{i = 0}^{{2^n} - 1} {{\theta _i}{C_T}\left| i \right\rangle } \\
 = {\theta _s}\left| t \right\rangle  + {\theta _t}\left| s \right\rangle  + \sum\limits_{i = 0,i \ne s,t}^{{2^n} - 1} {{\theta _i}\left| i \right\rangle }
\end{array}
\end{equation}
By comparing (\ref{eqn_19x7}) with Eq. (\ref{eqn_12}), we have proved that ${C_T} = {G_T}$, i.e., a sequence of quantum gates ${C^n}({X_y})$, $1 \le y \le n$, implements the two-point swapping ${G_T}$.

Since $s$ and $t$ differ in at most $n$ locations, there exists a Gray code such that $m \le n + 1$. In addition, the ${C^n}({X_y})$ gate is implemented using $O(n)$ single-qubit and controlled-NOT (i.e., ${N_{C1}}$ in Figure \ref{fig:2})) gates \cite{bibitem13, bibitem8}. ${C_T}$ is built by $2m - 3$ gates ${C^n}({X_y})$ ($1 \le y \le n$), thus the complexity of ${C_T}$ is $O({n^2})$, i.e., the complexity of a two-point swapping operator ${G_T}$ is also $O({n^{\rm{2}}})$.
\end{proof}
\section{Proof of Theorem \ref{theorem2}}
\begin{proof}
\label{proof2}
From definition \ref{def_4}, we see that a symmetric flip of a k-dimensional image with ${2^n}$ pixels requires a swapping of ${2^n}$ pixel points, thus the operator $G_F^{{v_i}}$ can be implemented using ${{\rm{2}}^{n{\rm{ - 1}}}}$ two-point swappings, i.e.,
\begin{equation}
\label{eqn_21x1}
T = \prod\limits_{i = 1}^{{2^{n - 1}}} {{G_{{T_i}}}}
\end{equation}
and
\begin{equation}
\label{eqn_21x2}
G_F^{\left| {{v_j}} \right\rangle }{\left| \psi  \right\rangle _k} = T{\left| \psi  \right\rangle _k}
\end{equation}
 where ${G_{{T_i}}}$ is the $i$th two-point swapping operator and  $\left| \psi \right\rangle_{k} $ is the NASS state in Eq. (\ref{eqn_3}). Since both operators are unitary, it is easy to see that

%We know that $G_F^{\left| {{v_j}} \right\rangle }$ and $T$ are unitary operators based on (\ref{eqn_11}), (\ref{eqn_20X2}), and (\ref{eqn_21x1}), so both sides of (\ref{eqn_21x2}) are multiplied by ${\left( {G_F^{\left| {{v_j}} \right\rangle }} \right)^ + }$ and ${T^ + }$ to become
%\begin{equation}
%\label{eqn_18}
%\left\{ \begin{array}{l}
%{(G_F^{\left| {{v_j}} \right\rangle })^{\rm{ + }}}G_F^{\left| {{v_j}} \right\rangle }{\left| \psi  \right\rangle _k} = {\left| \psi  \right\rangle _k} = {(G_F^{\left| {{v_j}} \right\rangle })^{\rm{ + }}}T{\left| \psi  \right\rangle _k}\\
%{T^{\rm{ + }}}G_F^{\left| {{v_j}} \right\rangle }{\left| \psi  \right\rangle _k} = {\left| \psi  \right\rangle _k} = {T^{\rm{ + }}}T{\left| \psi  \right\rangle _k}
%\end{array} \right.
%\end{equation}
%where ${\left( \cdot \right)^+}$ is the conjugate transpose of the operator.
%Thus, from (\ref{eqn_18}), we can infer that ${(G_F^{\left| {{v_j}} \right\rangle })^{\rm{ + }}}T = {T^{\rm{ + }}}G_F^{\left| {{v_j}} \right\rangle } = {I^{ \otimes n}}$, i.e.,
\begin{equation}
\label{eqn_21x3}
G_F^{\left| {{v_j}} \right\rangle }{\rm{ = }}T = \prod\limits_{i = 1}^{{2^{n - 1}}} {{G_{{T_i}}}}
\end{equation}
which implies that the operator $T$ can be implemented by ${{\rm{2}}^{n{\rm{ - 1}}}}$
%By analyzing (\ref{eqn_21x3}), we conclude that the implementation circuit of
two-point swappings (see Figure \ref{fig:8}). Therefore, the complexity of the symmetric flip is also $O(n)$.
\end{proof}


\section{Proof of Theorem \ref{theorem4}}
\begin{proof}
\label{proof4}
Set
\begin{equation}
\label{eqn_23x1}
\left| \psi  \right\rangle  = \sum\limits_{j = 0}^{{2^m} - 1} {{\alpha _j}\left| j \right\rangle }
\end{equation}

Application of ${T_{key}}$ to $\left| \psi \right\rangle $ gives that %in (\ref{eqn_23x1}) to obtain
\begin{equation}
\label{eqn_30}
{T_{key}}\left| \psi  \right\rangle  = {\alpha _{{2^m} - 1}}\left| 0 \right\rangle  + \sum\limits_{j = 0}^{{2^m} - 2} {{\alpha _j}} \left| {j + 1} \right\rangle
\end{equation}

Define a sequence of two-point swappings as follows.
\begin{equation}
\label{eqn_31}
\begin{array}{l}
G_T^{{2^m} - 1}:\left| {{2^m} - 1} \right\rangle  \leftrightarrow \left| 0 \right\rangle \\
G_T^{{2^m} - 2}:\left| {{2^m} - 2} \right\rangle  \leftrightarrow \left| {{2^m} - 1} \right\rangle \\
G_T^{{2^m} - 3}:\left| {{2^m} - 3} \right\rangle  \leftrightarrow \left| {{2^m} - 2} \right\rangle \\
 \vdots \\
G_T^1:\left| 1 \right\rangle  \leftrightarrow \left| 2 \right\rangle
\end{array}
\end{equation}

Applying $(G_T^{\rm{1}} \cdots G_T^{{2^m} - 2}G_T^{{2^m} - 1})$ to $\left| \psi \right\rangle $ in (\ref{eqn_23x1}), we obtain
\begin{equation}
\label{eqn_32}
(G_T^{\rm{1}} \cdots G_T^{{2^m} - 2}G_T^{{2^m} - 1})\left| \psi  \right\rangle  = {\alpha _{{2^m} - 1}}\left| 0 \right\rangle  + \sum\limits_{j = 0}^{{2^m} - 2} {{\alpha _j}} \left| {j + 1} \right\rangle
\end{equation}

It follows from (\ref{eqn_30}) and (\ref{eqn_32}) that $G_T^{\rm{1}}\cdots G_T^{{2^m} - 2}G_T^{{2^m} - 1} = {T_{key}}$, i.e., the operator ${T_{key}}$ in Eq. (\ref{eqn_22X10}) can be implemented by ${{\rm{2}}^{\rm{m}}}{\rm{ - 1}}$ two-point swappings.

By Theorem \ref{theorem1} we know that the complexity of $G_T^1\cdots G_T^{{2^m} - 2}G_T^{{2^m} - 1}$ is $O({2^m}{m^2})$, i.e., the complexity of ${T_{key}}$ is $O({2^m}{m^2})$. From Figure \ref{fig:t12} we see that the complexity of ${T_{\left| {{v_x}} \right\rangle }}$ is the same as that of ${T_{key}}$, i.e., the complexity of ${T_{\left| {{v_x}} \right\rangle }}$ is also $O({2^m}{m^2})$.
\end{proof}

\section*{Acknowledgments}
This work is supported by National Natural Science Foundation of China (Nrs. 61463016, 61462026, 11271138, 11531004), Simons Foundation Grant No. 198129, Program for New Century Excellent Talents in University under Grant No. NCET-13-0795, Landing project of science and technique of colleges and universities of Jiangxi Province under Grant No. KJLD14037, Training program of academic and technical leaders of Jiangxi Province under Grant No. 20153BCB22002, Natural Science Foundation of Jiangxi Province (No. 20151BAB207019), the advantages of scientific and technological innovation team of Nanchang City under Grant No. 2015CXTD003, the research funds of East China Jiaotong University (15XX02 and 15QT02), and an award of China Scholarship Council.



%% References




\end{document}